\documentclass{llncs}
\usepackage{tikz}
\usepackage{hyperref}
\usepackage{amsmath}
\usepackage{amssymb,stmaryrd}
\usepackage{listings}
\usepackage{verbatim}
\usepackage{color}
\usepackage{graphicx}
\usepackage{algorithm}%
\usepackage[noend]{algpseudocode}%
\usepackage{multirow}
\usepackage{pdfpages}
\usepackage{multibib}
\usepackage{wrapfig}
\usepackage{paralist}

\usetikzlibrary{shapes}
\usetikzlibrary{positioning}
\usetikzlibrary{arrows}
\usetikzlibrary{calc}
\usetikzlibrary{automata}
\usepackage{times}
\newcites{add}{Additional References}

\newcommand{\theoremlike}[2]{\par\medskip\penalty-250\refstepcounter{theorem}{{\bfseries\noindent#2
\ref{#1}.}}}

\newcommand{\thmhelperpre}[2]{\theoremlike{#1}{#2}}
\newcommand{\thmhelperpost}{\par\medskip}

\newenvironment{reflemma}[1]{\thmhelperpre{#1}{Lemma}}{\thmhelperpost}

\makeatletter
\renewcommand*\l@author[2]{}
\renewcommand*\l@title[2]{}
\makeatletter

\long\def\symbolfootnote[#1]#2{\begingroup%
\def\thefootnote{\fnsymbol{footnote}}\footnote[#1]{#2}\endgroup}

\newcommand{\straa}{\sigma}
\newcommand{\xu}{R}
\newcommand{\xg}{B}
\newcommand{\app}[1]{\mathit{Appear}(\omega,{#1})}
\newcommand{\makeTerminal}{\textsc{MakeTerminal}}
\newcommand{\collapse}{\textsc{Collapse}}

\newcommand{\maxE}{\mathit{E_m}}
\newcommand{\statebound}{K}

\newcommand{\otf}{\textsc{On-the-fly-EC}}

\newcommand{\pop}{pop}

\newcommand{\makestep}{\mathit{makestep}}

\newcommand{\strats}{\Sigma}

\newcommand{\support}[1]{\mathit{supp}(#1)}

\newcommand{\enab}{E}

\newcommand{\prmdp}[2]{\M{\otimes}\A}

\DeclareMathOperator*{\argmax}{arg\,max}

\newcommand{\sectref}[1]{Sec.~\ref{#1}}

\renewcommand{\algref}[1]{Algorithm~\ref{#1}}

\newcommand{\algalgref}[2]{Algorithms~\ref{#1} and \ref{#2}}

\newtheorem{assumption}{Assumption}

\newcounter{exampcount}
\newcommand{\startpara}[1]{{%
\vskip6pt\noindent
{\bf #1.}}}

\def\Nset{\mathbb{N}}

\def\Qset{\mathbb{Q}}
\def\Rset{\mathbb{R}}

\newcommand{\true}{\mathtt{true}} 
\newcommand{\false}{\mathtt{false}} 

\def\squareforqed{\hbox{\rlap{$\sqcap$}$\sqcup$}}
\def\qed{\ifmmode\squareforqed\else{\unskip\nobreak\hfil
\penalty50\hskip1em\null\nobreak\hfil\squareforqed
\parfillskip=0pt\finalhyphendemerits=0\endgraf}\fi}

\newcommand\mdp{{\sf M}}

\newcommand{\loos}{0}
\newcommand{\win}{1}
\newcommand{\learn}{\mathrm{Learn}}
\newcommand{\m}{m}
\newcommand{\eps}{\epsilon}
\newcommand{\states}{S}
\newcommand\sinit{{\bar{s}}}

\newcommand\act{A}
\newcommand\dist{{\mathit{Dist}}}
\newcommand\Dist{{\dist}}

\def\Prb{{\mathit{Pr}}}

\newcommand{\reach}[1]{\Diamond #1} 
\newcommand{\reachk}[1]{\Diamond_{\leq k} #1} 

\newcommand{\mdptuple}{{\langle \states,\sinit,\act, E, \Delta\rangle}}
\def\sinit{{\overline{s}}}

\newcommand{\M}{{\mathcal{M}}}

\newcommand{\last}{\mathit{last}}
\newcommand{\ipaths}{\mathit{IPath}}
\newcommand{\fpaths}{\mathit{FPath}}

\newcommand{\A}{\mathcal{A}}

\newcommand{\one}{1}
\newcommand{\zero}{0}
\newcommand{\accum}{\mathit{accum}_m}
\newcommand{\val}{V}

\newcommand{\maxU}{\mathit{maxU}}

\title{Verification of Markov Decision Processes \\ using Learning Algorithms}

\author{Tom\'a\v{s} Br\'{a}zdil\inst{1} \and Krishnendu Chatterjee\inst{2}\and Martin Chmel\'{i}k\inst{2}
 \and Vojt\v{e}ch~Forejt\inst{3} \and\\ Jan K\v{r}et\'insk\'y\inst{2}
 \and Marta~Kwiatkowska\inst{3}
 \and David~Parker\inst{4} \and Mateusz Ujma\inst{3}}

\authorrunning{Br\'azdil et al.}

\institute{Masaryk University, Brno, Czech Republic \and IST Austria
\and University of Oxford, UK \and University of Birmingham, UK}

\begin{document}
\maketitle

\begin{abstract}$\!\!\!$%
We present a general framework for applying machine-learning algorithms
to the verification of Markov decision processes (MDPs).
The primary goal of these techniques is to improve performance
by avoiding an exhaustive exploration of the state space.
Our framework focuses on probabilistic reachability, which is a core property for verification,
and is illustrated through two distinct instantiations.
The first assumes that full knowledge of the MDP is available,
and performs a heuristic-driven partial exploration of the model,
yielding precise lower and upper bounds on the required probability.
The second tackles the case where we may only sample the MDP,
and yields probabilistic guarantees, again in terms of
both the lower and upper bounds, which provides efficient stopping criteria for the approximation.
The latter is the first extension of statistical model-checking for unbounded properties in MDPs.
In contrast with other related approaches, we do not restrict our attention
to time-bounded (finite-horizon) or discounted properties,
nor assume any particular properties of the MDP.
We also show how our techniques extend to LTL objectives.
We present experimental results showing the performance of our framework on several examples.
\end{abstract}

\section{Introduction}

Markov decision processes (MDPs) are a widely used model for the formal verification
of systems that exhibit stochastic behaviour.
This may arise due to the possibility of failures (e.g.\ of physical system components),
unpredictable events (e.g.\ messages sent across a lossy medium),
or uncertainty about the environment (e.g.\ unreliable sensors in a robot).
It may also stem from the explicit use of randomisation,
such as probabilistic routing in gossip protocols
or random back-off in wireless communication protocols.

Verification of MDPs against 
temporal logics such as PCTL 
and LTL 
typically reduces to the computation of optimal (minimum or maximum) reachability probabilities,
either on the MDP itself or its product with some deterministic $\omega$-automaton.
Optimal reachability probabilities (and a corresponding optimal strategy for the MDP)
can be computed in polynomial time through a reduction to linear programming,
although in practice verification tools often use dynamic programming techniques,
such as value iteration which approximates the values up to
some pre-specified convergence criterion.

The efficiency or feasibility of verification
is often limited by excessive time or space requirements,
caused by the need to store a full model in memory.
Common approaches to tackling this include:
symbolic model checking, 
which uses efficient data structures
to construct and manipulate a compact representation of the model;
abstraction refinement, 
which constructs a sequence of increasingly precise approximations,
bypassing construction of the full model using decision procedures such as SAT or SMT;
and statistical model checking~\cite{YS02,HLMP04},
which uses Monte Carlo simulation to generate approximate results of verification
that 
hold with high probability.

In this paper, we explore the opportunities offered by learning-based methods,
as used in fields such as planning 
or reinforcement learning~\cite{SB98}.
In particular, we focus on algorithms that explore an MDP by generating trajectories through it and,
whilst doing so, produce increasingly precise approximations for some property of interest
(in this case, reachability probabilities). 
The approximate values, along with other information,
are used as heuristics to guide the model exploration
so as to minimise the solution time and the portion of the model that needs to be considered.


We present a general framework for applying such algorithms to the verification of MDPs.
Then, we consider two distinct instantiations that operate under different assumptions
concerning the availability of knowledge about the MDP, and produce different classes of results.
We distinguish between \emph{complete information}, where
full knowledge of the MDP is available (but not necessarily generated and stored),
and \emph{limited information}, where (in simple terms)
we can only sample trajectories of the MDP.

The first algorithm assumes complete information
and is based on \emph{real-time dynamic programming} (RTDP)~\cite{BBS95}.
In its basic form, this only generates approximations in the form of lower bounds
(on maximum reachability probabilities).
While this may suffice in some scenarios (e.g. planning),
in the context of verification we typically require more precise guarantees.
So we consider bounded RTDP (BRTDP)~\cite{MBLG05}, which supplements this with an additional upper bound.
%
%
The second algorithm assumes limited information
and is based on {\em delayed Q-learning} (DQL)~\cite{Strehl}.
Again, we produce both lower and upper bounds but,
in contrast to BRTDP, where these are guaranteed to be correct,
DQL offers probably approximately correct (PAC) results, i.e., there is a non-zero
probability that the bounds are incorrect.

Typically, MDP solution methods based on learning or heuristics make assumptions
about the structure of the model. For example, the presence of end components~\cite{DeA97}
(subsets of states where it is possible to remain indefinitely with probability 1)
can result in convergence to incorrect values.
Our techniques are applicable to arbitrary MDPs.
We first handle the case of MDPs that contain no end components
(except for trivial designated goal or sink states).
Then, we adapt this to the general case by means of \emph{on-the-fly}  detection of end components,
which is one of the main technical contributions of the paper.
We also show how our techniques extend to LTL objectives and thus also to minimum reachability probabilities.

Our DQL-based method, which yields PAC results, can be seen as an instance of
statistical model checking~\cite{YS02,HLMP04},
a technique that has received considerable attention.
Until recently, most work in this area focused on purely probabilistic models,
without nondeterminism, but several approaches have now been presented
for statistical model checking of nondeterministic models
\cite{DBLP:conf/formats/DavidLLMPVW11,DBLP:conf/cav/DavidLLMW11,DBLP:conf/ifm/Larsen13,DBLP:conf/forte/BogdollFHH11,LP12,DBLP:conf/qest/HenriquesMZPC12,DBLP:journals/corr/LegayS13}.
However, these methods all consider either time-bounded properties or use discounting to ensure convergence
(see below for a summary).
The techniques in this paper are the first
for statistical model checking of unbounded properties on MDPs.

We have implemented our framework within the PRISM tool~\cite{KNP11}.
This paper concludes with experimental results
for an implementation of our BRTDP-based approach
that demonstrate considerable speed-ups over the fastest methods in PRISM.

Detailed proofs omitted due to lack of space are available in
\ifthenelse{\isundefined{\techreport}}{\cite{arxiv}.}{the appendix.}


\vspace*{-0.8em}
\subsection{Related Work}

In fields such as planning and artificial intelligence, many learning-based
and heuristic-driven solution methods for MDPs have been developed.
In the \emph{complete information} setting, examples include RTDP~\cite{BBS95} and
BRTDP~\cite{MBLG05}, as discussed above, 
which generate lower and lower/upper bounds on values, respectively.
Most algorithms make certain assumptions in order to ensure convergence,
for example through the use of a discount factor
or by restricting to so-called Stochastic Shortest Path (SSP) problems,
whereas we target arbitrary MDPs without discounting.
%
More recently, an approach called FRET~\cite{KMWG11} was proposed 
for a generalisation of SSP, but this gives only a one-sided (lower) bound.
We are not aware of any attempts to apply or adapt such methods in the context of probabilistic verification.
A related paper is \cite{AL12}, which applies heuristic search methods to MDPs,
but for generating probabilistic counterexamples.


As mentioned above, in the \emph{limited information} setting,
our algorithm based on delayed Q-learning (DQL) 
yields PAC results, similar to those obtained from
\emph{statistical model checking}~\cite{YS02,HLMP04,DBLP:conf/cav/SenVA04}.
This is an active area of research with a variety of tools~\cite{DBLP:conf/tacas/JegourelLS12,DBLP:journals/corr/abs-1207-1272,DBLP:conf/qest/BoyerCLS13,DBLP:conf/mmb/BogdollHH12}.
In contrast with our work, most techniques focus on time-bounded properties,
e.g., using bounded LTL, rather than \emph{unbounded} properties.
Several approaches have been proposed to transform checking of unbounded properties
into testing of bounded properties, for example,
\cite{DBLP:conf/sbmf/YounesCZ10,DBLP:conf/kbse/HeJBGW10,SVA05,RP09}.
However, these focus on purely probabilistic models, without nondeterminism, and do not apply to MDPs.
In~\cite{DBLP:conf/forte/BogdollFHH11}, unbounded properties are analysed for MDPs
with spurious nondeterminism, where the way it is resolved does not affect the desired property.

More generally, the development of statistical model checking techniques
for probabilistic models with \emph{nondeterminism}, such as MDPs, is an important topic,
treated in several recent papers.
One approach is to give the nondeterminism a probabilistic semantics,
e.g., using a uniform distribution instead, as for timed automata
in \cite{DBLP:conf/formats/DavidLLMPVW11,DBLP:conf/cav/DavidLLMW11,DBLP:conf/ifm/Larsen13}.
Others~\cite{LP12,DBLP:conf/qest/HenriquesMZPC12},
like this paper, aim to quantify over all strategies and produce an $\epsilon$-optimal strategy.
The work in \cite{LP12} and \cite{DBLP:conf/qest/HenriquesMZPC12}
deals with the problem in the setting of discounted (and for the purposes of approximation thus bounded) or bounded properties, respectively.
In the latter work, candidates for optimal schedulers are generated and gradually improved, but ``at any given point we cannot quantify how close to optimal the candidate scheduler is'' (cited from~\cite{DBLP:conf/qest/HenriquesMZPC12}) and the algorithm 
``does not in general converge to the true optimum''
(cited from~\cite{DBLP:conf/sefm/LegayST14}). Further, \cite{DBLP:conf/sefm/LegayST14} considers compact representation of schedulers, but again focuses only on (time) bounded properties.

Since statistical model checking is simulation-based, one of the most important difficulties is the analysis of \emph{rare events}. This issue is, of course, also relevant for our approach; see the section on experimental results. Rare events have been addressed using methods such as importance sampling \cite{DBLP:conf/kbse/HeJBGW10,DBLP:conf/cav/JegourelLS12} and importance splitting \cite{DBLP:conf/cav/JegourelLS13}.

End components in MDPs can be collapsed
either for algorithmic correctness~\cite{DeA97}
or efficiency~\cite{CBGK08} (where only lower bounds on maximum reachability probabilities are considered).
Asymptotically efficient ways to detect them are given in \cite{CH11,CH12}.

\section{Basics about MDPs and Learning Algorithms}\label{sec:prelims}
We begin by giving some basic background material on MDPs
and establishing some fundamental definitions for our learning framework.

\startpara{Basic notions}
We use $\Nset$, $\Qset$, and $\Rset$
to denote the sets of all non-negative
integers, rational numbers, and real numbers,
respectively. Also, given real numbers $a\leq b$, we denote by $[a,b]\subseteq \Rset$ the closed interval between $a$ and $b$.
We assume familiarity with basic notions of probability theory, e.g.,
\emph{probability space} and \emph{probability measure}. %
As usual, a \emph{probability distribution} over a finite or
countably infinite set $X$ is a function
$f : X \rightarrow [0,1]$ such that \mbox{$\sum_{x \in X} f(x) = 1$}.
We call $f$ \emph{rational} if $f(x) \in
\Qset$ for every $x \in X$. %
We denote by $\support{f}$ the set of all $x\in X$ such that $f(x)>0$ and by $\Dist(X)$ the set of all rational probability distributions on $X$.

\subsection{Markov Decision Processes}
We work with \emph{Markov decision processes} (MDPs),
a widely used model to capture both nondeterminism
(for, e.g., control, concurrency) and probability.
\begin{definition}
An \emph{MDP}
is a tuple $\mdp=\mdptuple$, where $\states$ is a finite set of {\em states},
$\sinit\in S$ is an {\em initial state}, $\act$ is a finite set of {\em actions}, $\enab : S\rightarrow 2^A$ assigns
non-empty sets of {\em enabled} actions to all states, and $\Delta : S {\times} \act \rightarrow \Dist(S)$ is a (partial) {\em probabilistic transition function}
defined for all $s$ and $a$ where $a\in \enab(s)$.
\end{definition}

\begin{remark}
For simplicity of presentation we assume w.l.o.g. that, for every action $a \in \act$, there is at most
one state $s$ such that $a \in \enab(s)$, i.e., $\enab(s)\cap \enab(s')=\emptyset$ for $s\not = s'$.
If there \emph{are} states $s,s'$ such that $a \in \enab(s) \cap \enab(s')$, we can always rename
the actions as $(s,a) \in \enab(s)$, and $(s',a) \in \enab(s')$, so that the MDP satisfies our assumption.
\end{remark}

An {\em infinite path} %
of $\mdp$ is an infinite sequence $\omega=s_0 a_0 s_1 a_1 \cdots $
such that $a_i\in E(s_i)$ and $\Delta(s_i,a)(s_{i+1})>0$ for every $i\in \Nset$.
A {\em finite path} is a finite prefix of an infinite path ending in a state.
We use $\last(\omega)$ to denote the last state of a finite path $\omega$.
We denote by $\ipaths$ ($\fpaths$) the set of all infinite (finite) paths,
and by $\ipaths_s$ ($\fpaths_s$) the set of infinite (finite) paths starting in a state $s$.

A state $s$ is {\em terminal} if all actions $a\in \enab(s)$ satisfy $\Delta(s,a)(s)=1$.
An {\em end component} (EC) of $\mdp$ is a pair $(\states',\act')$ where
$\states'\subseteq \states$ and $\act' \subseteq \bigcup_{s \in \states'}E(s)$ such that:
(1) if $\Delta(s,a)(s') > 0$
 for some $s\in \states'$ and $a\in \act'$, then $s'\in \states'$; and
(2) for all $s,s'\in \states'$ there is a path $\omega=s_0 a_0 \ldots s_n$ such
 that $s_0=s$, $s_n=s'$ and for all $0\le i < n$
we have $a_i\in \act'$.
A {\em maximal end component} (MEC) is an EC that is maximal with respect to the point-wise subset ordering.

\startpara{Strategies}
A \emph{strategy} of MDP $\mdp$ is a function $\sigma:\fpaths \rightarrow \dist(\act)$ satisfying
$\support{\sigma(\omega)}\subseteq \enab(\last(\omega))$ for every $\omega\in \fpaths$.
Intuitively, the strategy resolves the choices of actions in each finite path by choosing
(possibly at random) an action enabled in the last state of the path.
We write $\strats_\mdp$ for the set of all strategies in $\mdp$.
In standard fashion \cite{KSK76},
a strategy $\sigma$ induces, for any state $s$, a probability measure
$\Prb_{\mdp,s}^{\sigma}$ over $\ipaths_s$.
A strategy $\sigma$ %
is \emph{memoryless} if $\sigma(\omega)$ depends only on $\last(\omega)$.

\startpara{Objectives and values}
Given a set $F\subseteq S$ of target states,
the \emph{bounded} reachability for step $k$, denoted by $\reachk{F}$, consists of
the set of all infinite paths that reach a state in $F$ within $k$ steps, and
the \emph{unbounded} reachability, denoted by $\reach{F}$, consists of the set of all infinite
paths that reach a state in $F$.
Note that $\reach{F}=\bigcup_{k\geq 0} \reachk{F}$.
We consider the \emph{reachability probability} $\Prb_{\mdp,s}^{\sigma}(\reach{F})$, and
strategies that maximise this probability.
We denote by $\val(s)$ the {\em value} in $s$,
defined by $\sup_{\sigma \in \strats_\mdp} \Prb_{\mdp,s}^{\sigma}(\reach{F})$.
Given $\epsilon\geq0$, we say that a strategy $\sigma$ is {\em $\epsilon$-optimal in $s$} if $\Prb_{\mdp,s}^{\sigma}(\reach{F})+\epsilon \geq \val(s)$,
and we call a $0$-optimal strategy {\em optimal}.
It is known~\cite{Puterman:book} that,
for every MDP, there is a memoryless optimal strategy.
We are interested in strategies that approximate the value function,
i.e., compute $\epsilon$-optimal strategies for $\epsilon>0$.
\subsection{Learning Algorithms for MDPs}
In this paper we study a class of learning-based algorithms that stochastically approximate the value function of a given MDP. Let us fix, for the whole section, an MDP $\mdp=\mdptuple$ and a set of target states $F$.

We denote by $V:S\times A\rightarrow [0,1]$ the {\em value function} for state-action pairs of $\mdp$,
defined for all $(s,a)$ where $s\in \states$ and $a\in \enab(s)$ by:
\[
V(s,a):=\sum_{s'\in S} \Delta(s,a)(s')\cdot \val(s').
\]
Intuitively, $V(s,a)$ is the value in $s$ assuming that the first action performed is $a$.
A {\em learning algorithm} $\mathcal{A}$ simulates executions of $\mdp$, and iteratively updates
 upper and lower approximations $U:S\times A\rightarrow [0,1]$ and $L:S\times A\rightarrow [0,1]$, respectively, of the true value function $V:S\times A\rightarrow [0,1]$.
 The simulated execution starts in the initial state $\sinit$. The functions $U$ and $L$ are initialized to appropriate values so that $L(s,a)\leq V(s,a)\leq U(s,a)$ for all $s\in S$ and $a\in A$. During the computation of $\mathcal{A}$, the simulated execution moves from state to state according to choices made by the algorithm and the values of $U(s,a)$ and $L(s,a)$ are updated for the states visited by the simulated execution.
The learning algorithm $\mathcal{A}$ terminates when $\max_a U(\sinit,a)-\max_a L(\sinit,a)<\epsilon$ where the {\em precision} $\epsilon>0$ is given to the algorithm as an argument.

As the values $U(s,a)$ and $L(s,a)$ are updated only for states $s$ visited, and possibly updated with new values that are based on the simulations,
the computation of the learning algorithm may be randomised and even give incorrect results with some probability.

\begin{definition}
Denote by $\mathcal{A}(\epsilon)$ the instance of learning algorithm $\mathcal{A}$ with precision $\epsilon$.
We say that $\mathcal{A}$ {\em converges (almost) surely} if for every $\epsilon>0$ the computation of $\mathcal{A}(\epsilon)$ (almost) surely terminates with $L(\sinit,a)\leq V(\sinit,a)\leq U(\sinit,a)$.
\end{definition}
In some cases almost-sure convergence cannot be guaranteed. In such cases we demand  that the computation terminates correctly with sufficiently high probability. In such a case we assume that the algorithm is given not only the precision $\epsilon$ but also an {\em error tolerance} $\delta>0$ as an argument.
\begin{definition}
Denote by $\mathcal{A}(\epsilon,\delta)$ the instance of learning algorithm $\mathcal{A}$ with precision $\epsilon$ and error tolerance $\delta$.
We say that $\mathcal{A}$ is {\em probably approximately correct (PAC)} if, for every $\epsilon>0$ and every $\delta>0$, with probability at least $1-\delta$ the computation of $\mathcal{A}(\epsilon,\delta)$ terminates with $L(\sinit,a)\leq V(\sinit,a)\leq U(\sinit,a)$.
\end{definition}
The function $U$ defines a memoryless strategy $\sigma_U$ which in every state $s$ chooses all actions $a$ maximising the value $U(s,a)$ over $\enab(s)$ uniformly at random.
The strategy $\sigma_U$ is used in some of the algorithms and also %
contributes to the output.

\begin{remark}
Note that if the value function is defined as the infimum over
strategies (as in~\cite{MBLG05}), then the strategy chooses actions to minimise the
lower value.
Since we consider the dual case of supremum over strategies, the choice
of $\sigma_U$ is to maximise the upper value.
\end{remark}
In order to design a proper learning algorithm, we have to specify what knowledge about the MDP $\mdp$ is available at the beginning of the computation. We distinguish the following two distinct %
cases:
\begin{definition}\label{def:limited-info}
A learning algorithm has {\em limited information} about $\mdp$ if it knows only
the initial state $\sinit$, a~number $\statebound \ge |S|$, a number $\maxE \geq\max_{s\in S}|\enab(s)|$, a number \[0<p_{\min} \le \min %
\{\Delta(s,a)(s')\mid s\in \states,a\in\enab,s'\in\support{\Delta(s,a)}\},\] and the function $\enab$
(more precisely, given a state $s$, the learning procedure can ask an oracle for $\enab(s)$).
We assume that the algorithm may simulate an execution of $\mdp$ starting with $\sinit$ and choosing enabled actions in individual steps.
\end{definition}
\begin{definition}
A learning algorithm has {\em complete information} about $\mdp$ if it knows the complete MDP $\mdp$.
\end{definition}

\label{sec:prelim}
\section{MDPs without End Components}\label{sec:nomec}

In this section we present algorithms for MDPs without ECs.
Let us fix an MDP $\mdp=\mdptuple$ and a target set $F$. 

\smallskip\noindent{\bf Assumption-EC.} 
We assume that $\mdp$ contains no ECs, with the 
exception of two  trivial ECs containing only distinguished terminal states $\one$ and $\zero$, respectively, with $F=\{\one\}$ and $V(\zero) = 0$.
This assumption (Assumption-EC) considerably simplifies the adaptation of BRTDP and DQL to the unbounded reachability objective. 
Later, in Section~\ref{sec:mec}, we show how to extend our methods to deal with arbitrary MDPs (i.e., MDPs with ECs).

We start by formalising our framework for learning algorithms outlined in the previous section,
Then, we instantiate the framework and obtain two learning algorithms: BRTDP and DQL.

\begin{algorithm}
\caption{Learning algorithm (for MDPs with no ECs)}\label{alg:skeleton}
\begin{algorithmic}[1]
\State \textbf{Inputs:} An EC-free MDP $\mdp$
\State $U(\cdot,\cdot)\gets1, L(\cdot,\cdot)\gets0$
\State $L(1,\cdot)\gets1, U(0,\cdot)\gets0$ \Comment \textsc{Initialise}
\Repeat  
 \noindent /* \textsc{Explore} phase */
 \State $\omega\gets \sinit$
 \Repeat \label{ln:explore-begin}
	\State $a\gets$ sampled uniformly from\ \ $\displaystyle\argmax_{a \in \enab(\last(\omega))} U(\last(\omega),a)$
    \State \label{ln:explore-state} $s\gets$ sampled according to $\Delta(\last(\omega),a)$ \Comment \textsc{GetSucc}($\omega,a$)
    \State $\omega\gets \omega\ a\,s$
 \Until $s\in\{1,0\}$ \Comment \textsc{TerminatePath}($\omega$) \label{ln:explore-end}
 
 \noindent /* \textsc{Update} phase */
 \Repeat
	\State $s'\gets \pop(\omega)$ 	
 	\State $a\gets \pop(\omega)$    
    \State $s\gets \last(\omega)$
      \State \textsc{Update}($(s,a),s'$)
 \Until $\omega=\sinit$ %
\Until $\displaystyle\max_{a \in \enab(\sinit)} U(\sinit,a)-\max_{a \in \enab(\sinit)} L(\sinit,a)<\epsilon$ \Comment {\textsc{Terminate}}
\end{algorithmic}
\end{algorithm}
\subsection{Our framework} 
Our framework is presented as Algorithm~\ref{alg:skeleton}, 
and works as follows. 
Recall that functions $U$ and $L$ store the current upper and lower bounds on the value function $\val$, respectively.
Each iteration of the outer loop is divided into two phases: \textsc{Explore} and \textsc{Update}. 
In the \textsc{Explore} phase (lines 5 - 10), the algorithm samples a finite path $\omega$ 
in $\mdp$ from $\sinit$ to a state in $\{\one,\zero\}$ by always randomly choosing one of the enabled actions
that maximise the $U$ value, and sampling the successor state using the probabilistic transition function. 
 In the \textsc{Update} phase (lines 11 - 16), the 
 algorithm updates $U$ and $L$ on the state-action pairs along the path in a 
 backward manner. Here, the function $\pop$ pops and returns the last letter of the given sequence.

\subsection{Instantiations: BRTDP and DQL}
Our framework will be instantiated with two different algorithms, and
the difference between them is the way that the \textsc{Update} function is defined.

\subsubsection{Unbounded reachability with BRTDP.}

We obtain BRTDP by instantiating \textsc{Update} with Algorithm~\ref{alg:RTDP},
which requires complete information about the MDP.
Intuitively, \textsc{Update} computes new values of $U(s,a)$ and $L(s,a)$ by taking the weighted average of corresponding $U$ and $L$ values, respectively, over all successors of $s$ via the action $a$.
Formally, denote $\displaystyle U(s)=\max_{a\in\enab(s)}U(s,a)$ and $\displaystyle L(s)=\max_{a\in\enab(s)}L(s,a)$.
\begin{algorithm}
\caption{BRTDP\label{alg:RTDP} instantiation of Algorithm \ref{alg:skeleton}}
\begin{algorithmic}[1]
\Procedure {Update}{$(s,a),\_$}
  \State $U(s,a):=\sum_{s'\in S}\Delta(s,a)(s')U(s')$
  \State $L(s,a):=\sum_{s'\in S}\Delta(s,a)(s')L(s')$
\EndProcedure  
\end{algorithmic}
\end{algorithm}
\begin{theorem}\label{thm:BRTDP-nomec}
The algorithm BRTDP converges almost surely 
under Assumption-EC.
\end{theorem}

\begin{remark}
Note that in the \textsc{Explore} phase, an action maximising the value of $U$ is chosen and the successor is sampled according to the probabilistic transition function of $\mdp$. However, one can consider various modifications. Actions and successors may be chosen in different ways (e.g., for \textsc{GetSucc}), 
for instance, uniformly at random, in a round-robin fashion, or assigning various probabilities (bounded from below by some fixed $p >0$) to all possibilities in any biased way. 
In order to guarantee almost-sure convergence, some conditions have to be satisfied. Intuitively we require, that the state-action pairs used by $\epsilon$-optimal strategies have to be chosen enough times. If the condition is satisfied then the almost-sure convergence is preserved and the practical running times may significantly improve. For details, see Section~\ref{sec:expt}.
\end{remark}

\begin{remark}
The previous BRTDP algorithm is only applicable if the transition probabilities are known.
However, if complete information is not known, but $\Delta(s,a)$ can be repeatedly sampled for all $s$ and $a$,
then a variant of BRTDP can be shown to be probably approximately correct.
\end{remark}

\noindent{\bf Unbounded reachability with DQL.}
On many occasions neither the complete information is available nor repeated sampling is possible, and we have to deal with only limited information about $\mdp$ (see Definition~\ref{def:limited-info}). 
\begin{algorithm}
\caption{DQL (delay $m$, estimator precision $\bar\epsilon$)\label{alg:DQL} instantiation of Algorithm \ref{alg:skeleton}}
\begin{algorithmic}[1]
\Procedure {Update}{$(s,a),s'}$ %
  \If {$c(s,a)=m$ \textbf{and} LEARN($s,a$)
  }
  \If{$\accum^U(s,a)/m < U(s,a) -2\bar\epsilon$}
    \State $U(s,a)\gets\accum^U(s,a)/m +\bar\epsilon$
	\State $\accum^U(s,a)=0$
  \EndIf
  \If{$\accum^L(s,a)/m > L(s,a) +2\bar\epsilon$}
    \State $L(s,a)\gets\accum^L(s,a)/m -\bar\epsilon$
	\State $\accum^L(s,a)=0$
  \EndIf
	\State $c(s,a)=0$
  \Else 
	\State $\accum^U(s,a)\gets \accum^U(s,a) + U(s')$
	\State $\accum^L(s,a)\gets \accum^L(s,a) + L(s')$
	\State $c(s,a)\gets c(s,a) +1$
  \EndIf
\EndProcedure  
\end{algorithmic}
\end{algorithm}

\vskip5pt
For this scenario, we use DQL, which can be obtained by instantiating \textsc{Update} with Algorithm~\ref{alg:DQL}. Here the macro LEARN($s,a$) is true in $k$th call of \textsc{Update}$((s,a),\cdot)$ if since the $(k-2m)$th call of \textsc{Update}$((s,a),\cdot)$ line 4 was not executed in any call of \textsc{Update}$(\cdot,\cdot)$.

The main idea behind DQL is as follows. As the probabilistic transition function is not known, we cannot update $U(s,a)$ and $L(s,a)$ with the precise values $\sum_{s'\in S}\Delta(s,a)(s')U(s')$ and $\sum_{s'\in S}\Delta(s,a)(s')L(s')$, respectively. However, we can use simulations executed in the \textsc{Explore} phase of Algorithm~\ref{alg:skeleton}  to estimate these values instead. Namely, we use $\accum^U(s,a)/m$ to estimate  $\sum_{s'\in S}\Delta(s,a)(s')U(s')$ where $\accum^U(s,a)$ is the sum of the $U$ values of the last $m$ immediate successors of $(s,a)$ seen during the \textsc{Explore} phase. Note that the delay $m$ must be chosen large enough for the estimates to be sufficiently close, i.e. $\bar\epsilon$-close, to the real values.

So in addition to $U(s,a)$ and $L(s,a)$, the algorithm uses new variables $\accum^U(s,a)$ and $\accum^L(s,a)$ to accumulate $U(s,a)$ and $L(s,a)$ values, respectively, and a counter $c(s,a)$ counting the number of invocations of $a$ in $s$ since the last update (all these variables are initialized to $0$ at the beginning of computation).
Assume that $a$ has been invoked in $s$ during the \textsc{Explore} phase of Algorithm~\ref{alg:skeleton}, which means that \textsc{Update}$((s,a),s')$ is eventually called in the Update phase of Algorithm~\ref{alg:skeleton} with the corresponding successor $s'$ of $(s,a)$. If $c(s,a)=m$ at the moment, $a$ has been invoked in $s$ precisely $m$ times since the last update concerning $(s,a)$ and the procedure \textsc{Update}$((s,a),s')$ updates $U(s,a)$ with 
$\accum^U(s,a)/m$ plus an appropriate constant $\bar\epsilon$ (unless LEARN is false). Here the purpose of adding $\bar\epsilon$ is to make $U(s,a)$ stay above the real value $V(s,a)$ with high probability. If $c(s,a)<m$, then \textsc{Update}$((s,a),s')$ simply accumulates the $U(s')$ into $\accum^U(s,a)$ and increases the counter $c(s,a)$. The $L(s,a)$ values are estimated by $\accum^L(s,a)/m$ in a similar way, just subtracting $\bar\epsilon$ from $\accum^L(s,a)$. The procedure requires $m$ and $\bar\epsilon$ as inputs, and they are chosen depending on $\epsilon$ and $\delta$, 
more precisely, we choose 
$\bar\epsilon=\frac{\eps \cdot (p_{\min} / \maxE)^{|\states|}}{12 |\states|}$
and 
$m= \frac{\ln(6 |\states||\act| ( 1 + \frac{|\states||\act|}{\bar\eps})/\delta)}{2 {\bar\eps}^2}$
and establish that DQL is probably approximately correct.
The parameters $m$ and $\bar{\eps}$ can be conservatively approximated using only the limited information about the MDP.
Even though the algorithm has limited information about $\mdp$, we still establish the following theorem.

\begin{theorem}\label{thm:dql}
DQL is probably approximately correct under Assumption-EC.
\end{theorem}

\begin{remark}[Bounded reachability]
Algorithm~\ref{alg:skeleton} can be trivially adapted to handle bounded reachability properties by preprocessing the input
MDP in a standard fashion. Namely, every state is equipped with a bounded
counter with values ranging from $0$ to $b$ where $b$ is the step bound, the current value
denoting the number of steps taken so far.
All target states remain target for all counter values, and every non-target state with counter value $b$ becomes rejecting.
Then to determine the $b$-step reachability in the original MDP one can compute the (unbounded) reachability in the new MDP.
Although this means that the number of states is multiplied by $b+1$, in practice
the size of the explored part of the model can be small.
\end{remark}

\section{Unrestricted MDPs}
\label{sec:mec}

\begin{wrapfigure}[10]{r}{4.5cm}
\vspace*{-3em}
\begin{tikzpicture}[x=1cm,y=1cm,outer sep=1pt]
\node[state,initial,initial text=] (a) {$m_1$}; 
\node[state] (b) [right=1 of a] {$m_2$};
\node (fork) [below=0.5 of b,outer sep=0,inner sep=0] {$\bullet$};
\node[state] (c) [below=1 of b] {$\one$};
\node[state] (d) [left=1 of c] {$\zero$};
\path[->] 
(a) edge[bend right] node[below]{$a\ \ 1$} (b)
(b) edge[bend right] node[above]{$b\ \ 1$} (a)
(fork) edge node[right]{$0.5$} (c)
(fork) edge[bend right=20pt] node[left,pos=0.7,xshift=-2pt]{$0.5$} (d)
(c) edge[loop right] node[pos=0.3,yshift=3pt]{$d$} node[pos=0.7]{$1$} ()
(d) edge[loop right] node[pos=0.3,yshift=1pt]{$e$} node[pos=0.7]{$1$} ()
;
\draw (b) edge node[right] {$c$} (fork);
\end{tikzpicture}
\vspace*{-1em}
\caption{MDP $\mdp$ with an EC.}\label{fig:noConverge}
\vspace*{-2em}
\end{wrapfigure}
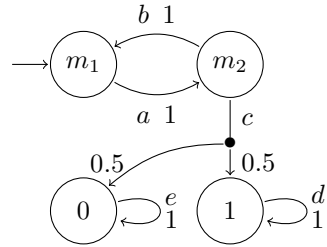

We first illustrate with an example that the algorithms BRTDP and DQL as presented 
in Section~\ref{sec:nomec} may not converge when there are ECs in the MDP.

\begin{example}
Consider the MDP~$\mdp$ on Figure~\ref{fig:noConverge} with an EC
 $(\xu,\xg)$, where $\xu = \{m_1,m_2\}$ and $\xg= \{a,b\}$.
While the values in states $m_1$ and $m_2 $ are $V(m_1)=V(m_2)=0.5$, we have for every iteration
 that the upper bound $U(m_1)=U(m_2)=1$. This follows from the fact that $U(m_1,a)=U(m_2,b)=1$ and both
 the algorithms greedily choose the action with the highest upper bound. It follows that in every iteration~$t$
 of the algorithm the error for the initial state $m_1$ is $U(m_1) - V_t(m_1) = \frac{1}{2}$ and the algorithm does not converge.

In general all states in an EC have the upper bound $U$ always equal to $1$ as by definition there is a non-empty
set of actions that are guaranteed to keep the next reached state in the EC, i.e., state with an upper bound of $1$.
 This argument holds even for the standard value iteration of the upper bounds initialized to $1$. 
\end{example}

One way of dealing with general MDPs is to preprocess an MDP and identify all MECs~\cite{CH11,CH12}, and 
``collapse'' every MEC into a single state (see e.g.~\cite{DeA97,CBGK08}). 
These algorithms require that the graph model is known and explore the whole state space, 
which might not be possible either due to limited information (see Definition~\ref{def:limited-info}), 
or because the model is too large. Hence, we propose a modification to the algorithms from the previous 
sections that allows us to deal with ECs ``on the fly''.
We first describe the collapsing of a set of states and then present a crucial lemma that allows us to
identify ECs to collapse.

\smallskip\noindent{\bf Collapsing states.}
In the following, we say that an MDP $\mdp'=\langle \states',\sinit',\act', E', \Delta'\rangle$
is obtained from $\mdp = \mdptuple$ by {\em collapsing} a tuple $(\xu,\xg)$, where
$\xu \subseteq S$ and $\xg \subseteq \act$ with $\xg  \subseteq \bigcup_{s \in \xu}E(s)$ if:
\begin{itemize}
 \item $\states' = (\states \setminus \xu) \cup \{s_{(\xu,\xg)}\}$,
 \item $\sinit'$ is either $s_{(\xu,\xg)}$ or $\sinit$, depending on whether $\sinit \in R$ or not,
 \item $\act' = \act \setminus \xg$,
 \item $\enab'(s) = \enab(s) \setminus \xg$, for $s \in \states \setminus \xu$; $\quad \enab'(s_{(\xu,\xg)}) = \bigcup_{s \in \xu} \enab(s) \setminus \xg$, 
 \item $\Delta'$ is defined for all $s\in \states'$ and $a\in E'(s)$ by
  \begin{itemize}
   \item $\Delta'(s,a)(s') = \Delta(s,a)(s')$ for $s,s'\neq s_{(\xu,\xg)}$.
   \item $\Delta'(s,a)(s_{(\xu,\xg)}) = \sum_{s'\in \xu} \Delta(s,a)(s')$ for $s\neq s_{(\xu,\xg)}$.
   \item $\Delta'(s_{(\xu,\xg)},a)(s') = \Delta(s,a)(s')$ for $s'\neq s_{(\xu,\xg)}$ and
	$s$ the %
        state with $a\in \enab(s)$.
   \item $\Delta'(s_{(\xu,\xg)},a)(s_{(\xu,\xg)}) = \sum_{s'\in \xu}\Delta(s,a)(s')$ where
	$s$ is the %
        state with $a{\in} \enab(s)$.
  \end{itemize}
\end{itemize}
We denote the above transformation as the \collapse{} function, i.e., 
$\collapse(\xu,\xg)$ creates $\mdp'$ from $\mdp$.
As a special case, given a state $s$ and a terminal state $s' \in \{0,1\}$ we use $\makeTerminal(s,s')$  as shorthand for \collapse{}$(\{s,s'\}, \enab(s))$ where the resulting
state is renamed to $s'$. Intuitively every transition leading to state $s$ will lead to the terminal state $s'$ in the modified MDP after $\makeTerminal(s,s')$.

For practical purposes, it is important to note that the
collapsing does not need to be implemented explicitly, but can be done by keeping a separate data structure 
which stores information about the collapsed states.

\smallskip\noindent{\bf Identifying ECs from simulations.}
Our modifications will identify ECs ``on-the-fly'' through simulations 
that get stuck in them. 
In the next lemma we establish the identification procedure.  
To this end, for a path $\omega$, let us denote by $\app{i}$ the tuple $(S_i,A_i)$ of $\mdp$ such 
that $s\in S_i$ and $a\in A_i(s)$ if and only if $(s,a)$ occurs in $\omega$ more than $i$ times.
\begin{lemma}\label{lemma:mec-explore}
Let $c=\exp\left(-\left(p_{\min}/\maxE\right)^{\kappa}\,/\,\kappa\right)$ where $\kappa=K\maxE+1$. 
Let $i\geq \kappa$. Assume that the \textsc{Explore} phase in Algorithm~\ref{alg:skeleton} terminates with probability less than $1$. Then, provided the \textsc{Explore} phase does not terminate within $3 i^3$ iterations, the conditional probability that $\app{i}$ is an EC is at least $1-2c^i i^3\cdot \left(p_{\min}/\maxE\right)^{-\kappa}$.
\end{lemma}

The above lemma allows us to modify the \textsc{Explore} phase of Algorithm~\ref{alg:skeleton} in such a way that simulations will be used to identify ECs. 
Discovered ECs will be subsequently collapsed.
We first present the overall skeleton (Algorithm~\ref{alg:skel-mec}) of the ``on-the-fly" ECs, which consists of two parts:
(i)~identification of ECs; and (ii)~processing them. 
The instantiations for BRTDP and DQL will differ in the identification 
phase.
Since the difference will be in the identification phase, before 
proceeding to the individual identification algorithms we first establish 
the correctness of the processing of ECs.

\begin{algorithm}[h]
\caption{Extension for general MDPs\label{alg:skel-mec}}
\begin{algorithmic}[1]
\Function{On-the-fly-EC}{}
\State $\mathcal{M} \gets \textsc{IdentifyECs} $ \Comment \textsc{Identification of ECs}
\ForAll{$(\xu,\xg)\in \mathcal{M}$ } \Comment \textsc{Process ECs}
 \State \collapse $(\xu,\xg)$ 
 \ForAll{$s\in \xu$ and $a\in E(s)\setminus \xg$}
  \State $U(s_{(\xu,\xg)},a) \gets U(s,a)$
  \State $L(s_{(\xu,\xg)},a) \gets L(s,a)$
 \EndFor
\If{$\xu \cap F \neq \emptyset$} 
\State \makeTerminal $(s_{(\xu,\xg)},\one)$
\ElsIf{no actions enabled in $s_{(\xu,\xg)}$} 
\State \makeTerminal $(s_{(\xu,\xg)},\zero)$ 
\EndIf
\EndFor
\EndFunction
\end{algorithmic}
\end{algorithm}

\begin{lemma}
\label{lem:proc}
Assume $(\xu,\xg)$ is an EC in MDP $\mdp$,  $V_\mdp$ the value before the \textsc{Process ECs} procedure in Algorithm~\ref{alg:skel-mec}, and $V_{\mdp'}$ the value after the procedure, then:
\begin{compactitem}
\item For $i \in \{0,1\}$ if $\makeTerminal(s_{(\xu,\xg)},i)$  is called, then $ \forall s\in \xu: \ V_{\mdp}({s})=i$;
\item $\forall s \in \states \setminus \xu: \  V_\mdp(s) = V_{\mdp'}(s);$
\item $\forall s \in \xu:  \  V_\mdp(s) = V_{\mdp'}(s_{(\xu,\xg)}).$
\end{compactitem}
\end{lemma}

\subsection{Complete information}

\noindent{\bf Modification of Algorithm~\ref{alg:skeleton}.}
To obtain BRTDP working with unrestricted MDPs, we modify Algorithm~\ref{alg:skeleton} as follows:
for iteration $i$ of the \textsc{Explore} phase, we insert a check after line 9 that if 
the length of the path $\omega$ (which is the number of states in $\omega$) explored is $k_i$, 
then we invoke the \otf{} function for BRTDP. 
The \otf{} function possibly modifies the MDP by processing (collapsing) some ECs 
as described in Algorithm~\ref{alg:skel-mec}.
After the \otf{} function terminates, we interrupt the current \textsc{Explore} phase, 
and start the \textsc{Explore} phase for the $i+1$-th iteration (i.e., 
generating a new path again, starting from $\sinit$ in the modified MDP).
To complete the algorithm description we describe the choice of $k_i$ and 
identification of ECs.

\smallskip\noindent{\bf Choice of $k_i$.}
We do not call \otf{} every time a new state is explored, but only after
every $k_i$ steps of the repeat-until loop at lines~\ref{ln:explore-begin}--\ref{ln:explore-end} in 
iteration $i$.
The specific value of $k_i$ can be decided experimentally and change as the computation progresses,
a reasonable choice for $k_i$ to ensure that there is an EC with high probability can be 
obtained from Lemma~\ref{lemma:mec-explore}.

\smallskip\noindent{\bf Identification of ECs: Algorithm~\ref{alg:brtdp-mec}.}
Given the current explored path $\omega$, let $(T,G)$ be $\app{0}$, that is
the set of states and actions explored in $\omega$.
To obtain the ECs from the set $T$ of explored states, 
Algorithm~\ref{alg:brtdp-mec} computes an auxiliary MDP $\mdp^T=\langle T',\sinit,\act', E', \Delta'\rangle$ 
defined as follows:
\begin{itemize}
 \item $T' = T\cup \{t \mid \exists s\in T, a\in \enab(s)\text{ such that }\Delta(s,a)(t) > 0\}$,
 \item $\act' = \bigcup_{s\in T} \enab(s) \cup \{a_\bullet\} $
 \item $E'(s)= E(s)$ if $s\in T$ and $E'(s) = \{a_\bullet\}$ otherwise
 \item $\Delta'(s,a) = \Delta(s,a)$ if $s\in T$, and $\Delta'(s,a_\bullet)(s) = 1$ otherwise. 
\end{itemize}
Then the algorithm computes all MECs of $\mdp^T$ that are contained in $T$ and 
identifies them as ECs.
The following lemma establishes that every EC identified is indeed one 
in the original MDP.
\begin{algorithm}
\caption{Identification of ECs for BRTDP\label{alg:brtdp-mec}}
\begin{algorithmic}[1]
\Function{\textsc{IdentifyECs}}{$\mdp$, $T$}
\State compute $\mdp^T$
\State $\mathcal{M}' \gets \text{ MECs of } \mdp^T$
\State $\mathcal{M} \gets \{(\xu,\xg)\in \mathcal{M}'\mid \xu \subseteq T \}$
\EndFunction
\end{algorithmic}
\end{algorithm}
\begin{lemma}\label{lem:mec-is-ec}
Let $\mdp$ and $\mdp^T$ be the MDPs from the construction above and $T$ the corresponding set of explored states.
 Then every MEC $(\xu,\xg)$ in $\mdp^T$ such that $\xu \subseteq T$ is an EC in $\mdp$.
\end{lemma}

Finally, we establish that the modified algorithm, which we refer as OBRTDP (on-the-fly BRTDP),
almost surely converges; and the proof is an extension of Theorem~\ref{thm:BRTDP-nomec}. 

\begin{theorem}\label{thm:BRTDP-mec}
The OBRTDP (on-the-fly BRTDP) algorithm converges almost surely for all MDPs. %
\end{theorem}
\subsection{Limited information}
We now present the on-the-fly algorithm for DQL. 
The three key aspects to describe are as follows: (i)~modification of 
Algorithm~\ref{alg:skeleton} and identification of ECs;
(ii)~interpretation of collapsing of ECs; and (iii)~the correctness 
argument.

\smallskip\noindent{\bf Modification Algorithm~\ref{alg:skeleton} 
and identification of ECs.}
The modification of Algorithm~\ref{alg:skeleton} is done exactly as 
for the modification of BRDTP (i.e., we insert a check after line 9
of \textsc{Explore}, which invokes the \otf\ function if the length of
path $\omega$ exceeds $k_i$). 
In iteration $i$, we set $k_i$ as $3 \ell_i^3$, for some $\ell_i$ 
(to be described later). 
The identification of the EC is as follows: 
we consider $\app{\ell_i}$ the set of states and actions that have appeared
more than $\ell_i$ times in the explored path $\omega$, which is of length 
$3\ell_i^3$, and identify the set as an EC;
i.e., $\mathcal{M}$ in line 2 of Algorithm~\ref{alg:skel-mec} is defined
as the set containing the single tuple $\app{\ell_i}$.
We refer the algorithm as ODQL (on-the-fly DQL).

\smallskip\noindent{\bf Interpretation of collapsing.}
We now describe the interpretation of collapsing for MDPs when the 
algorithm has limited information. 
Intuitively, once an EC $(\xu,\xg)$ is collapsed, the algorithm in 
the \textsc{Explore} phase can choose a state $s\in \xu$ and action $a \in E(s) \setminus \xg$
to leave the EC. 
This is simulated in the \textsc{Explore} phase by considering all actions of the 
EC uniformly at random until $s$ is reached, and then action $a$ is chosen.
Since $(\xu,\xg)$ is an EC, playing all actions of $B$ uniformly at random ensures
that $s$ is almost surely reached. Note, that the steps made inside a collapsed EC
do not count to the length of the explored path.

\smallskip\noindent{\bf Choice of $\ell_i$ and correctness.}
The choice of $\ell_i$ is as follows. 
Note that in iteration $i$, the error probability, obtained from Lemma~\ref{lemma:mec-explore},
is at most $2c^{\ell_i} \ell_i^3 \cdot \left(p_{\min}/\maxE\right)^{-\kappa}$ and we choose $\ell_i$ such that 
$2c^{\ell_i} \ell_i^3 \cdot \left(p_{\min}/\maxE\right)^{-\kappa} \leq \frac{\delta/2}{2^{i}}$, 
where $\delta$ is the error tolerance.
Note that since $c<1$, we have that $c^{\ell_i}$ decreases exponentially, 
and hence for every $i$ such $\ell_i$ exists. 
It follows that the total error of the algorithm due to the on-the-fly EC
collapsing is at most $\delta/2$.
It follows from the proof of Theorem~\ref{thm:dql}
that for ODQL the error is at most $\delta$ if we use the same $\bar\epsilon$ as for DQL, but now with DQL error tolerance $\delta/4$, i.e. with $m= \frac{\ln(24 |\states||\act| ( 1 + \frac{|\states||\act|}{\bar\eps})/\delta)}{2 {\bar\eps}^2}$.
\begin{theorem}
\label{thm:odql}
ODQL (on-the-fly DQL) is probably approximately correct for all MDPs. %
\end{theorem}

\newcommand{\RQ}{\mathit{M}}
\newcommand{\RS}{\mathit{N}}
\newcommand{\ov}{\overline}

\subsection{Extension to LTL}
While in this work we focus on probabilistic reachability, our techniques also extend to 
analysis of MDPs with linear temporal logic (LTL) objectives.
Given an LTL objective describing a set of desired infinite paths, the objective can be converted to 
deterministic $\omega$-automaton~\cite{DBLP:conf/lics/VardiW86,DBLP:conf/focs/Safra88,KE12,CGK13}, and thus analysis
of MDPs with LTL objectives reduces to analysis of MDPs with $\omega$-regular condition
such as Rabin acceptance conditions~\cite{Thomas96languages}.
A Rabin acceptance condition consists of a set $\{ (\RQ_1,\RS_1) \ldots (\RQ_d,\RS_d)\}$
of $d$ pairs $(\RQ_i,\RS_i)$, where each $\RQ_i \subseteq S$ and $\RS_i 
\subseteq S$.
The acceptance condition requires that for some $1\leq i\leq d$ states in $\RQ_i$ are visited
infinitely often and states in $\RS_i$ are visited finitely often. 

The value computation for MDPs with Rabin objectives is achieved as follows: an EC $(\xu,\xg)$ is 
\emph{winning} if for some $1\leq i \leq d$ we have $\xu \cap \RQ_i \neq \emptyset$ and 
$\xu \cap \RS_i =\emptyset$, and the value computation reduces to probabilistic reachability to winning
ECs~\cite{DBLP:conf/icalp/CourcoubetisY90}.
Thus extension of our results from reachability to Rabin objectives requires processing of ECs for 
Rabin objectives (line 3-11 of Algorithm~\ref{alg:skel-mec}). 
The principle of processing ECs is as follows: 
Given an EC $(\xu,\xg)$ is identified, we first obtain the EC in the original MDP (i.e., obtain the set of states and actions 
corresponding to the EC in the original MDP) as $(\ov{\xu},\ov{\xg})$ and then determine if there is a sub-EC of $(\ov{\xu},\ov{\xg})$ that is winning using standard algorithms for MDPs with Rabin objectives~\cite{DBLP:books/daglib/0020348}; 
and if so then we merge the whole EC as in line 9 of Algorithm~\ref{alg:skel-mec}; if not, and moreover, there is no action out of the EC, we merge as in line 11 of Algorithm~\ref{alg:skel-mec}.
With the modified EC processing we obtain OBRTDP and ODQL for MDPs with Rabin objectives.

\section{Experimental Results}\label{sec:expt}

\startpara{Implementation}
We developed an implementation of our learning-based framework
as an extension of the PRISM model checker~\cite{KNP11},
building upon its simulation engine for generating trajectories.%
We focus on the complete-information case, i.e., BRTDP,
for which we can perform a more meaningful comparison with PRISM.
We implement \algalgref{alg:skeleton}{alg:RTDP},
and the on-the-fly EC detection algorithm of \sectref{sec:mec},
with the optimisation of taking $T$ as the set of all states explored so far.

We consider three distinct variants of the learning algorithm,
by modifying the \textsc{GetSucc} function in~\algref{alg:skeleton},
which is the heuristic responsible for picking a successor state $s'$
after choosing some action $a$ in each state $s$ of a trajectory.
The first variant takes the unmodified \textsc{GetSucc}, selecting
$s'$ at random according to the distribution $\Delta(s,a)$. This behaviour
follows the one of the original RTDP algorithm~\cite{BBS95}.
The second uses the heuristic proposed for BRTDP in \cite{MBLG05},
selecting the successor $s'\in\support{\Delta(s,a)}$ that maximises the
difference $U(s')-L(s')$ between bounds for those states.
For the third, we propose an alternative approach
that systematically chooses all successors $s'$
in a round-robin (R-R) fashion, and guarantees sure termination.

\startpara{Results}
We evaluated our implementation on four existing benchmark models, using %
a machine with a 2.8GHz Xeon processor and 32GB of RAM, running Fedora~14.
We use three models from the PRISM benchmark suite~\cite{KNP12b}:
\emph{zeroconf}, %
\emph{wlan}, %
and \emph{firewire\_impl\_dl}; %
and a fourth one from \cite{FKP11}: \emph{mer}.
The first three use (unbounded) probabilistic reachability properties;
the fourth a time-bounded property.
The latter is used to show differences
between heuristics that were less visible in the unbounded case.
%
%

%
%
%

%
%
%
%

We run BRTDP and compare its performance to PRISM.
We terminate it when the bounds $L$ and $U$ differ by at most $\epsilon$
for the initial state of the MDP. We use $\epsilon=10^{-6}$ in all cases
except \emph{zeroconf}, where $\epsilon=10^{-8}$ is used since the actual values are very small.
For PRISM, we use its fastest engine, which is the ``sparse'' engine,
running value iteration. This is terminated when the values for all states
in successive iterations differ by at most $\epsilon$.
Strictly speaking, this is not guaranteed to produce an $\epsilon$-optimal strategy
(e.g. in the case of very slow numerical convergence), but on all these examples it does.

The experimental results are summarised in Table \ref{heuristics_table}.
For each model, we give the
the number of states in the full model and the time for PRISM
(model construction, precomputation of zero/one states and value iteration)
and for BRTDP with each of the three heuristics described earlier.
All times have been averaged over 20 runs.

\begin{table}[!t]
\vspace*{-3em}
\begin{center}	
\begin{tabular}{|c|c|r||c|c|c|c|}
\hline
\multirow{2}{1.6cm}{\centering Name\\{[param.s]}} & 
\multirow{2}{2.0cm}{\centering $\,$Param.$\,$\\values} & 				
\multirow{2}{1.8cm}{\centering Num.\\states} &
\multicolumn{4}{c|}{Time (s)} \\\cline{4-7}
& & &
{\centering PRISM}&
{\centering RTDP}&
{\centering BRTDP}&
{\centering $\ $R-R$\ $}\\
\hline
\hline
\multirow{3}{2.1cm}{\centering \emph{zeroconf}\\$[N, K]$}
 & $20, 10$ & 3,001,911 & 129.9 & 7.40 & 1.47 & 1.83\\
 & $20, 14$ & 4,427,159 & 218.2 & 12.4 & 2.18 & 2.26\\
 & $20, 18$ & 5,477,150 & 303.8 & 71.5 & 3.89 & 3.73\\

\hline
\multirow{3}{2.1cm}{\centering \emph{wlan}\\$[BOFF]$}
 & $4$ & 345,000 & 7.35 & 0.53 & 0.48 & 0.54\\
 & $5$ & 1,295,218 & 22.3 & 0.55 & 0.45 & 0.54\\
 & $6$ & 5,007,548 & 82.9 & 0.50 & 0.43 & 0.49\\
\hline

\multirow{3}{2.4cm}{\centering \emph{firewire\_impl\_dl}\\$[delay,$\\$deadline]$}
 & $36, 200$ & 6,719,773 & 63.8 & 2.85 & 2.62 & 2.26\\
 & $36, 240$ & 13,366,666 & 145.4 & 8.37 & 7.69 & 6.72\\
 & $36, 280$ & 19,213,802 & 245.4 & 9.29 & 7.90 & 7.39\\
\hline

\multirow{4}{2.4cm}{\centering \emph{mer}\\$[N, q]$}
 & $3000, 0.0001$ & 17,722,564 & 158.5 & 67.0 & 2.42 & 4.44\\
 & $3000, 0.9999$ & 17,722,564 & 157.7 & 10.9 & 2.82 & 6.80\\
 & $4500, 0.0001$ & 26,583,064 & 250.7 & 67.3 & 2.41 & 4.42\\
 & $4500, 0.9999$ & 26,583,064 & 246.6 & 10.9 & 2.84 & 6.79\\

\hline
\end{tabular}
\end{center}
\caption{Verification times using BRTDP (three different heuristics) and PRISM.}
\label{heuristics_table}
\end{table}

We see that our method outperforms PRISM on all four benchmarks.
The improvements in execution time on these benchmarks are possible because
the algorithm is able to construct an $\epsilon$-optimal policy
whilst exploring only a portion of the state space.
The number of distinct states visited by the algorithm is, on average,
three orders of magnitude smaller that the total size of the model
(column `Num. states') and reachable state space under the optimal adversary
contains hundreds of states.

The RTDP heuristic is generally the slowest of the three,
and tends to be sensitive to the probabilities in the model.
In the \emph{mer} example, changing the parameter $q$
can mean that some states, which are crucial for the convergence of the algorithm,
are no longer visited due to low probabilities on incoming transitions.
This results in a considerable slow-down.
This is a potential problem for MDPs containing rare events i.e. modelling failures that
occur with very low probability.
The BRTDP and R-R heuristics perform very similarly, despite being quite
different (one is randomised, the other deterministic).
Both perform consistently well on these examples.

\vspace{-0.5em}

\section{Conclusions}

We have presented a framework for verifying MDPs using learning algorithms.
Building upon methods from the literature, we provide novel techniques to analyse
unbounded probabilistic reachability properties of arbitrary MDPs,
yielding either exact bounds, in the case of complete information,
or probabilistically correct bounds, in the case of limited information.
Given our general framework, one possible direction would be to explore
other learning algorithms in the context of verification.
Another direction of future work is to explore whether learning algorithms
can be combined with symbolic methods for probabilistic verification.\\

\paragraph{Acknowledgement} We thank Arnd Hartmanns and anonymous reviewers for careful reading and valuable feedback.

\bibliographystyle{splncs03}
\bibliography{smc2}

\newpage

\appendix

\section{Proof of Theorem~\ref{thm:BRTDP-nomec}: Correctness of BRTDP}\label{app:mecs}
Assume that there are no ECs in $\mdp$
with the exception of (trivial) components containing two distinguished terminal states $\one$ (the only target state) and $\zero$ (a ``sink'' state).
Consider Algorithm~\ref{alg:skeleton} with \textsc{Update} defined in Algorithm~\ref{alg:RTDP}, but now with line 17 being ``\textbf{until} false'', i.e. iterating the outer repeat loop ad infinitum. Denote the functions $U$ and $L$ after $i$ iterations by $U_i$ and $L_i$, respectively.

\begin{lemma}\label{lem:monotone-brtdp-nomec}
For every $i\in\mathbb N$, all $s\in S$ and $a\in A$, 
\[U_1(s,a)\geq \cdots \geq U_i(s,a)\geq V(s,a) \geq L_i(s,a)\geq \cdots\geq  L_1(s,a)\]
\end{lemma}
\begin{proof}
Simple induction. \qed
\end{proof}

\begin{lemma}
$\lim_{i\to\infty}(U_i(\sinit)-L_i(\sinit))=0$ almost surely.
\end{lemma}
\begin{proof}
Let $a_i(s)\in \enab(s)$ maximise $U_i(s,a)$ and define $\delta_i(s):= U_i(s,a_i(s))-L_i(s,a_i(s))$. Since $\delta_i(s)\geq \max_a U_i(s,a)-\max_a L_i(s,a)$ (expression of line 17 in the original Algorithm~\ref{alg:skeleton}),
it is sufficient to prove that $\lim_{i\to\infty}\delta_i(\sinit)=0$ almost surely.

By Lemma~\ref{lem:monotone-brtdp-nomec}, the limits $\lim_{i\to\infty} U_i(s,a)$ and $\lim_{i\to\infty} L_i(s,a)$ are well defined and finite. Thus $\lim_{i\to\infty}\delta_i(s)$ is also well defined and we denote it by $\delta(s)$ for every $s\in S$.

Let $\Sigma_U$ be the set of all memoryless strategies in $\mdp$ which occur as $\sigma_{U_i}$ for infinitely many $i$. Each $\sigma\in \Sigma_U$ induces a chain with reachable state space $\states_{\sigma}$ and uses actions $\act_{\sigma}$. Note that under $\sigma\in \Sigma_U$, all states of $\states_{\sigma}$ will be almost surely visited infinitely often if infinitely many simulations are run. Similarly, all actions of $\act_{\sigma}$ will be used almost surely infinitely many times. Let $\states_{\infty}=\bigcup_{\sigma\in\Sigma_U} \states_{\sigma}$ and let $\act_{\infty}=\bigcup_{\sigma\in \Sigma_U} \act_{\sigma}$. During almost all computations of the learning algorithm, all states of $\states_{\infty}$ are visited infinitely often, and all actions of $\act_{\infty}$ are used infinitely often.  By definition of $\delta$, for every $t\in \states_{\infty}$ and $a\in \act_{\infty}$ holds $\delta(t)=\sum_{s\in \states_{\infty}} \Delta(t,a)(s)\cdot \delta(s)$ almost surely.

Let $\delta=\max_{s\in S_{\infty}}\delta(s)$ and $D=\{s\in S_{\infty}\mid \delta(s)=\delta\}$. To obtain a contradiction, consider a computation of the learning algorithm such that $\delta>0$ and $\delta(t)=\sum_{s\in \states_{\infty}} \Delta(t,a)(s)\cdot \delta(s)$ for all $s\in S$ and $a\in \enab(s)$. Then $\one,\zero\not\in D$ and thus
$D$ cannot contain any EC by assumption. By definition of EC we get
$$\exists t\in D:\forall a\in \enab(t):supp(\Delta(t,a))\not\subseteq D$$
and thus for every $a\in\enab(t)$ we have $t_a\notin D$ with $\Delta(t,a)(t_a)>0$. %
Since $t_a\notin D$ we have $\delta(t_a)<\delta$. Now for every $a\in \enab(t)\cap \act_{\infty}$ we have:
\begin{align*}
\delta(t)&=\sum_{s\in S_{\infty},s\neq t_a}\Delta(t,a)(s)\cdot\delta(s)+\Delta(t,a)(t_a)\cdot \delta(t_a)\\
&<\sum_{s\in S_{\infty},s\neq t_a}\Delta(t,a)(s)\cdot\delta+\Delta(t,a)(t_a)\cdot \delta\\
&=\delta
\end{align*}
a contradiction with $t\in D$. 

\qed
\end{proof}

As a corollary, Algorithm 1 with \textsc{Update} defined in Algorithm 2 almost surely terminates for any $\varepsilon>0$. Further, $U_i\geq V\geq L_i$ pointwise and invariantly for every $i$ by the first lemma, the returned result is correct.

\newcommand{\wb}{\overline}
\newcommand{\reals}{\mathbb{R}}
\newcommand{\close}{K}
\newcommand{\trans}{\Delta}
\newcommand{\expect}{E}

\newcommand{\set}[1]{\{#1\}}
\newcommand{\pacStraa}{\mathcal{A}}
\newcommand{\greStraa}{\pi}
\newcommand{\reached}{\mathcal{R}}
\newcommand{\cons}{\overline{t}}

\section{Proof of Theorem \ref{thm:dql}: Analysis of the DQL algorithm}\label{app:dql}

In this section we present the analysis of the DQL algorithm for MDPs with reachability objectives, and show that 
the algorithm is probably approximately correct. We explicitly provide Algorithm~\ref{alg:pac} as a full pseudocode of the DQL algorithm, with minor modifications that will be discussed later.

\smallskip\noindent\textbf{Initialization of Algorithm~\ref{alg:pac}.}
The algorithm initializes the following variables:
$U(s,a)$ is the upper bound on the value of the state action pair $(s,a)$ and is initialized to 0 for $s=0$ and to $1$ otherwise; $\accum^U$ is the accumulator as discussed in Section~\ref{sec:nomec}, and is initialized
to $0$; $c^U(s,a)$ is counting the number of times the state action pair $(s,a)$ was experienced, and is initialized to $0$; $t^u(s,a)$ is the iteration number (timestep) of the last update 
of the $U(s,a)$ estimate of the state action pair $(s,a)$, initialized to 0; and $Learn^U(s,a)$ is boolean flag indicating whether the strategy
 is considering a modification to its upper-bound estimate $U(s,a)$; and the value $t^U_*$ denotes the iteration (timestep) of the last upper bound estimate change, and is initialized to $0$.
  The similar variables for lower bounds are distinguished by a $L$ superscript.

\smallskip\noindent\textbf{Body of Algorithm~\ref{alg:pac}.}
Let $s$ denote the state of the MDP in iteration (timestep) $t$. In every iteration the algorithm chooses uniformly at random an action $a$ from the set of enabled actions $\enab(s)$, that has
a maximal estimate of the upper bound. The strategy plays action $a$ and the MDP reaches a new state $s'$. If the strategy considers updating of the $U(s,a)$ estimate, the value
$U(s')$ is added to the estimator $\accum^U(s,a)$. Whenever the state action pair $(s,a)$ is experienced $m$ times, an attempt to update the estimate $U(s,a)$ will occur. The update will be successful
if the difference between the current estimate $U(s,a)-\accum^U(s,a)/m$ is greater or equal to $2\eps_1$. In case of a successful update the new upper bound for the state action
pair $(s,a)$ is $\accum^U(s,a)+\eps_1$ (the precise values for $m$ and $\eps_1$ will be given later in the analysis part). If the attempted update is not successful and $t^U(s,a) \geq t^U_*$
the strategy will not consider any updates of the upper bound until some other state action pair $(s',a')$ is successfully updated. The code for lower bound estimates is symmetric with a single
difference: when the strategy does not intend to perform updates of the lower bound  of the state action pair $(s,a)$, i.e., $\learn^L(s,a) =\false$ a successful update of the upper bound estimate can make the strategy consider 
the updates again, i.e., sets $\learn^L(s,a)$ back to $\true$.
Finally, if the newly reached state $s'$ is in $\{\zero, \one\}$, i.e.,  the simulation reached a terminal state and is restarted back to the initial state $\sinit$. Otherwise the following iterations starts
with $s$ being the newly reached state $s'$.

For simplicity of analysis we consider a slightly less succinct version of the DQL algorithm presented in the main text (Algorithm~\ref{alg:skeleton} with Algorithm~\ref{alg:DQL} as
the \textsc{Update} function) and present Algorithm~\ref{alg:pac}. The main differences of the algorithms are as follows:
\begin{compactitem}
\item 
For every state-action pair $(s,a)$ we introduce separate boolean flag $\learn^U(s,a)$ (resp. $\learn^L(s,a)$) for the upper (resp. lower) bound. The DQL procedure in Algorithm~\ref{alg:DQL}
contains a single shared boolean flag. Having these flags separated allows us to reason about upper bounds without considering updates of the lower bounds. The effect of having a single shared flag does not affect the bounds presented in Theorem~\ref{thm:dql}. This follows from the fact that only the number of attempted updates of the upper bound estimates has doubled. However, as the number of steps 
is given in big O notation, the statement of Theorem remains unaffected.
\item The updates of the $\accum^U$ (resp. $\accum^L$) accumulator are performed in a different order. In Algorithm~\ref{alg:pac} the updates of the accumulator are executed immediately. In the DQL procedure
of Algorithm~\ref{alg:DQL} the updates of the accumulator occur only after the simulation reaches the terminal state ($\zero$ or $\one$) and the updates are done in a stack-like fashion, i.e., last visited
state is the first to be updated. First consider an intermediate step: the updates of the accumulator are performed in a queue-like fashion, after the simulation reaches the terminal states. This can double the amount of required iterations, as for an update to be performed one has to wait until the simulation terminates. However, the constant does not affect the statement of Theorem~\ref{thm:dql} as the results are given in big O notation. It is easy to observe, that updating values in a stack-like fashion can only increase the rate of convergence as opposed to queue like updating. This follows from the fact that every simulation ends in terminal state $\zero$ or $\one$ and propagating the  value in a stack-like fashion can update the accumulator even of the initial state after a single simulation.
\end{compactitem}

\begin{algorithm}[]
\caption{DQL algorithm}
\label{alg:pac}
\begin{algorithmic}[1]
\State \textbf{Inputs:} $((\states,\loos,\win),\act,\m,\eps_1)$
\ForAll{$(s,a) \in \states \times \act$}
\State $U(s,a) \gets 1$; $\accum^U(s,a) \gets 0$; $c^U(s,a) \gets 0$; $t^U(s,a) \gets 0$; $\learn^U(s,a) \gets \true$
\State $L(s,a) \gets 0$; $\accum^L(s,a) \gets 0$; $c^L(s,a) \gets 0$; $t^L(s,a) \gets 0$; $\learn^L(s,a) \gets \true$
\EndFor
\ForAll{$(s,a) \in 0 \times \act$}
\State $U(s,a) \gets 0$
\EndFor
\ForAll{$(s,a) \in 1 \times \act$}
\State $L(s,a) \gets 1$
\EndFor

\State $t^U_* \gets 0$; $t^L_* \gets 0$;\ $s \gets s_0$

\For{$t=1,2,3, \ldots$}
\State Choose uniformly an action $a$ from  $\argmax_{a' \in \enab(s)} U(s,a')$
\State $s' \gets \makestep(s,a)$
 
\Comment{Upper Bounds:}
\If{$\learn^U(s,a)$}
\State  $\accum^U(s,a) \gets \accum^U(s,a) + U(s')$
\State $c^U(s,a) \gets c^U(s,a) +1$
\If{$c^U(s,a) = m$}
\Comment Update attempt
\If{$U(s,a) - \accum^U(s,a)/m \geq 2 \eps_1$}

\State $U(s,a) \gets \accum^U(s,a)/m + \eps_1$
\Comment Successful attempt
\State $t^U_* \gets t$
\ElsIf{$t^U(s,a) \geq t^U_*$}
 $\learn^U(s,a) \gets \false$
\EndIf
\State $t^U(s,a) \gets t;\  \accum^U(s,a) \gets 0;\  c^U(s,a) \gets 0$
\EndIf
\ElsIf{$t^U(s,a) < t^*$} $\learn^U(s,a) \gets \true$

\EndIf

\Comment{Lower Bounds:}

\If{$\learn^L(s,a)$}
\State  $\accum^L(s,a) \gets \accum^L(s,a) + L(s')$
\State $c^L(s,a) \gets c^L(s,a) +1$
\If{$c^L(s,a) = m$}
\Comment Update attempt
\If{$\accum^L(s,a)/m - L(s,a) \geq 2 \eps_1$}

\State $L(s,a) \gets \accum^L(s,a)/m - \eps_1$
\Comment Successful attempt
\State $t^L_* \gets t$
\ElsIf{$t^L(s,a) \geq t^L_*$}
 $\learn^L(s,a) \gets \false$
\EndIf
\State $t^L(s,a) \gets t;\  \accum^L(s,a) \gets 0;\  c^L(s,a) \gets 0$
\ElsIf{$t^L(s,a) < \max(t^L_*,t^U_*$)}
 $\learn^L(s,a) \gets \true$
 \EndIf

\EndIf

\If{$s' \in 0 \cup 1$ }
\Comment{Terminate the simulation or continue}
 \State $s' \gets s_0$
\Else
 \State$s \gets s'$
\EndIf

\EndFor
\end{algorithmic}
\end{algorithm}

\smallskip\noindent\textbf{Analysis of upper bounds $U(s,a)$.}
In what follows we present an adapted proof of DQL~\cite{Strehl} that analyses MDPs with discounted rewards. We adapt the proof to our setting of
undiscounted reachability objectives.
We write $U_t(s) = \max_{a \in \act}\{U_t(s,a)\}$ for the maximal U-value estimate of state $s$ at iteration (time) $t$. We write $U^*(s)$ for the actual upper bound
at state $s$ and $U^*(s,a)$ for the actual upper bound at state $s$ when action $a$ is played. We denote by $U^{\straa}_{\mdp}(s,T)$ the value function of 
strategy $\straa$ in MDP $\mdp$ starting in state $s$ for the $T$-step bounded reachability objective.

\begin{assumption}
\label{ass:correctMDP}
We say an MDP $\mdp$ with a reachability objective $\reach(F)$ satisfies Assumption~\ref{ass:correctMDP} if for every state $s \neq \zero$ that is in an EC, we have
that the value of that state is~$1$.
\end{assumption}

\begin{lemma}
\label{lem:perturbation}
Let $\mdp$ be an MDP and $\sigma$ a memoryless strategy
ensuring that a terminal state is reached almost surely.
Then the system of Bellman equations
\begin{eqnarray*}
 f(1) &=& 1\\
 f(0) &=& 0\\
 f(s) &=& \varepsilon(s) + \sum_{s'\in S} \sigma(s)(a)\cdot\Delta(s,a)(s') \cdot f(s') \quad\text{otherwise}
\end{eqnarray*}
 has unique solution, for any choice of numbers $\varepsilon(s)$.
\end{lemma}
\begin{proof}
Let $F:\Rset^{|S|} \rightarrow \Rset^{|S|}$ be the function
that performs one iteration of the Bellman equations.
We show that $F^{|S|}$ is a contraction. Let $P(s,s',k)$ be the probability
that when using $\sigma$ and starting in $s$, we end in $s'$ after exactly $k$ steps.
{\small
\begin{eqnarray*}
 \lefteqn{F^{|S|}(x)(s) }\\
 &=& \Big(\sum_{s'\in S}P(s,s',|S|) \cdot x(s')\Big)
    + \Big(\sum_{s'\in S\setminus\{0,1\}} \sum_{i=0}^{|S|-1}P(s, s', i) \cdot \varepsilon(s')\Big)\\
 &=& \Big(\sum_{s'\in S \setminus \{0,1\}}P(s, s', |S|) \cdot x(s')\Big) + P(s, 1, |S|)
  + \Big(\sum_{s'\in S\setminus\{0,1\}} \sum_{i=0}^{|S|-1}P(s, s', i) \cdot \varepsilon(s')\Big)
\end{eqnarray*}
}%
note that in the last line above, the second and third summands are independent of $x$, and so
{\small
\begin{eqnarray*}
 \lefteqn{|F^{|S|}(x)(s) -  F^{|S|}(y)(s)|}\\
  &=& \Big|\Big(\sum_{s'\in S \setminus \{0,1\}}P(s, s', |S|) \cdot x(s')\Big)
  -\Big(\sum_{s'\in S \setminus \{0,1\}}P(s, s', |S|) \cdot y(s')\Big)\Big|\\
  &=& \sum_{s' \in S \setminus \{0,1\}}(P(s, s', |S|) \cdot |x(s') - y(s')|\\
  &\le& ||(x,y)||_\infty \cdot\sum_{s' \in S \setminus \{0,1\}}(P(s, s', |S|)
\end{eqnarray*}
}%
where $||(x,y)||_\infty = \max_s |x(s) - y(s)|$ is the maximum norm.
Because $\mdp$ is MEC-free, a terminal state is reached with nonzero probability within $|S|$ steps,
and hence $\sum_{s'\in S \setminus \{0,1\}}(P(s, s', |S|) < 1$, implying that
\[
  || F^{|S|}(x), F^{|S|}(y)||_\infty<||(x,y)||_\infty.
\]
We have proved that $F^{|S|}$ is a contraction.
Applying Banach fixpoint theorem we get that there is a unique fixpoint for $F^{|S|}$, meaning that
$f$ has unique solution.
\end{proof}

\begin{lemma}
Let $\mdp$ be a MEC-free MDP, then $\mdp$ satisfies Assumption~\ref{ass:correctMDP}.
\end{lemma}
\begin{proof}
Trivially, by definition.
\end{proof}

\begin{lemma}
\label{lem:suc_updates}
The number of successful updates of the U-value estimates in Algorithm~\ref{alg:pac} is bounded by $\frac{\vert\states\vert \vert\act\vert}{\eps_1}$.
\end{lemma}
\begin{proof}
Let $(s,a) \in \states \times \act$ be a fixed pair and $U(s,a)$ its value estimate. The value of $U(s,a)$ is initialized to $0$ or $1$
 and every successful update decreases the estimate by at least $\eps_1$. It is also impossible for any update to result in a negative U-value 
 estimate. It follows that the number of successful estimate updates for a fixed pair $(s,a)$ is bounded by $\frac{1}{\eps_1}$. As there
 are $ \vert\states\vert   \vert\act\vert $ pairs of U-value estimates for the upper bounds in the algorithm, we have that the number of successful updates
 of U-value estimates is bounded by $ \frac{\vert\states\vert \vert\act\vert}{\eps_1}$.
 \hfill
 \qed
\end{proof}

\begin{lemma}
\label{lem:att_updates}
The number of attempted updates of the U-value estimates in Algorithm~\ref{alg:pac} is bounded by $ \vert\states\vert \vert\act\vert( 1+ \frac{\vert\states\vert \vert\act\vert}{\eps_1})$.
\end{lemma}
\begin{proof}
Let $(s,a) \in \states \times \act$ be a fixed pair, $U(s,a)$ its value estimate, and $\learn^U(s,a)$ the learning flag for this pair. After
visiting state $s$ and playing action $a$ for $m$ times an attempt to update the estimate $U(s,a)$ occurs. In order to have another attempt to update the estimate
$U(s,a)$ some some estimate needs to be successfully updated after the last attempt to update the estimate $U(s,a)$. If there is no successful update of any estimate,
 the learning flag $\learn^U(s,a)$ is set to $\false$, and there are no attempted updates for $U(s,a)$ while $\learn^U(s,a) = \false$. By Lemma~\ref{lem:suc_updates} the
 number of successful updates of the U-value estimates in bounded by $\frac{\vert\states\vert \vert\act\vert}{\eps_1}$. It follows that the number of attempted updates
 for the estimate $U(s,a)$ is bounded by $1+\frac{\vert\states\vert \vert\act\vert}{\eps_1}$. As there
 are $ \vert\states\vert   \vert\act\vert $ many pairs of U-value estimates in the algorithm, we have that the number of attempted updates
 of U-value estimates is bounded by $ \vert\states\vert \vert\act\vert( 1+ \frac{\vert\states\vert \vert\act\vert}{\eps_1})$.
 \hfill
 \qed
\end{proof}

For every timestep $t$ we define $\close_t$ to be the set of all state-action pairs $(s,a) \in \states \times \act$ such that:
$$U_t(s,a) - \sum_{s' \in \states} \trans(s,a)(s') U_t(s') \leq 3 \eps_1 $$

\begin{assumption}
\label{ass:1}
Suppose an attempted update of the U-value estimate $U(s,a)$ of the pair $(s,a) \in \states \times \act$ occurs at time $t$, and that the $m$ most 
recent visits to the state $s$ while $a$ was played are $k_1 < k_2 < \cdots < k_m  = t$. If $(s,a) \not \in \close_{k_1}$, then the attempted update 
at time $t$ will be successful.
\end{assumption}
We specify the value of $m$:
$$m = \frac{\ln(6 \vert\states\vert \vert\act\vert ( 1 + \frac{\vert\states\vert \vert\act\vert}{\eps_1})/\delta)}{2 {\eps_1}^2}$$
\begin{lemma}
\label{lem:a1Holds}
The probability that Assumption~\ref{ass:1} is violated during the execution of Algorithm~\ref{alg:pac} is bounded by $\delta /6$.
\end{lemma}
\begin{proof}
Fix any timestep $k_1$ (and the complete history up to time $k_1$) such that at time $k_1$ state $s$ is visited and action $a$ played,
$(s,a) \not \in \close_{k_1}$,  and after 
$m-1$ more visits to state $s$ while playing action $a$ an attempt of an update will occur.
Let $\mathcal{Q} = \langle s[1],s[2], \ldots , s[m] \rangle \in \states^m$ be any sequence of $m$ next states, reached from state $s$ after playing action $a$.
Due to the Markov property, whenever the strategy is in a state $s$ and plays action $a$, the resulting next state does not depend on the 
history of the play. Therefore, the probability that the state $s$ is visited and action  $a$ is played  $m-1$ more times and the resulting sequence
of next states is equal to $\mathcal{Q}$, is at most the probability that $\mathcal{Q}$ is obtained by $m$ independent draws from the transition probability
distribution $\trans(s,a)$. It follows that it suffices to show that the probability that a random sequence $\mathcal{Q}$ causes an unsuccessful update
is at most $\delta/3$.

Let us fix a sequence of states $\mathcal{Q} = \langle s[1],s[2], \ldots , s[m]\rangle$ that is drawn from the transition probability distribution $\trans(s,a)$.
Let $X_1, X_2, \ldots, X_m$ be a sequence of independent and identically distributed random variables, where every $X_i$ is defined as $X_i = U_{k_1}(s[i])$. 
Let $\wb{X} = \frac{1}{m}\sum_{i=1}^{m} X_i$,
by Hoeffding bound~\cite{H63} we have the following inequality:
$$\Prb(\wb{X} - \expect[\wb{X}] < \eps_1) > 1 - e^{-2m {\eps_1}^2} $$
Our choice of $m$ evaluates the right-hand side of the inequality to 
$$1 - e^{-2m {\eps_1}^2} = 1- \frac{\delta}{(6 \vert\states\vert \vert\act\vert (1 + \frac{\vert\states\vert \vert\act\vert}{\eps_1}))}$$

As the random variables are independent and identically distributed we have that $\expect[\wb{X}] = \expect[X_i]$ for all $1 \leq i \leq m$, in particular 
$\expect[\wb{X}] = \expect[X_1]$. It follows that the probability of $\wb{X} - \expect[X_1] < \eps_1$ is at least $1- \frac{\delta}{(6 \vert\states\vert \vert\act\vert
 (1 + \frac{\vert\states\vert \vert\act\vert}{\eps_1}))}$.

If $\wb{X} - \expect[X_1] < \eps_1$ holds and an attempt to update the U-value estimate of the pair $(s,a)$ occurs using these $m$ samples, the update will be successful. Suppose
that state $s$ is visited and action $a$ played at times $k_1 < k_2 < \ldots < k_m$, where $k_m = t$ and at time $k_i$ the next state drawn from the transition
distribution $\trans(s,a)$ is $s[i]$. Then we have:
\begin{align*}
U_t(s,a) - \frac{1}{m}\sum_{i=1}^{m}U_{k_i}(s[i]) &\geq U_t(s,a) - \frac{1}{m}\sum_{i=1}^{m}U_{k_1}(s[i])\\
& > U_t(s,a) - \expect[X_1] - \eps_1\\
& = U_{k_1}(s,a) - \sum_{s' \in \states} \trans(s,a)(s') U_{k_1}(s') - \eps_1\\
&> 2\eps_1
\end{align*}

The first inequality follows from the fact, that the U-value estimates can only decrease, i.e., for all states $s \in \states$ and all $i \leq j$ we have $U_{i}(s) \geq U_j(s)$. The 
second inequality follows from the presented Hoeffding bound. The equality follows from the definition and the fact that $U_t(s,a) = U_{k_m}(s,a) = U_{k_1}(s,a)$,\
 and the last inequality follows from the assumption that $(s,a) \not \in \close_{k_1}$,
i.e., $U_{k_1}(s,a) - \sum_{s' \in \states} \trans(s,a)(s') U_{k_1}(s') > 3 \eps_1$.

To conclude the proof we extend the argument, using the union bound, to all possible timesteps $k_1$ that satisfy the conditions above. The number of such timesteps is
bounded by the number of attempted updates, that is by Lemma~\ref{lem:att_updates} equal to $\vert\states\vert \vert\act\vert (1+ \frac{\vert\states\vert \vert\act\vert}{\eps_1})$. We have that the probability
that Assumption~\ref{ass:1} is violated is at most $\delta/6$.

\hfill
 \qed
\end{proof}

\begin{lemma}
\label{lem:uppBound}
During the execution of Algorithm~\ref{alg:pac} we have that $U_t(s,a) \geq U^*(s,a)$ for all timesteps $t$ and state action pairs $(s,a)$ is at most $1 - \delta/6$.
\end{lemma}
\begin{proof}
It can be shown by similar arguments as in Lemma~\ref{lem:a1Holds}, that $1/m \sum_{i=1}^{m}U^{*}(s_{k_i}) \geq U^*(s,a) - \eps_1$  holds, for all attempted
updates, with probability at least $1 - \delta/6$. Assuming the inequality holds, the proof is by induction on the timestep $t$. For the base case we initialize
all U-value estimates for states $s$ in $\states \setminus \zero$ and actions $a \in \act$ to $U_1(s,a)=1$ which is clearly an upper bound. All the states $s \in \zero$
are absorbing non-target states, therefore the initialization value $0$ is also an upper bound.
Suppose the claim holds for all timesteps less than or equal to $t$, i.e., $U_t(s,a) \geq U^*(s,a)$ and $U_t(s) \geq U^*(s)$ for all state action pairs $(s,a)$.

Assume $s$ is the $t$-th state reached and $a$ is an action played at time $t$. If there is no attempt to update or the update is not successful, no U-value estimate is changed 
and there is nothing to prove. Assume there was successful update of the U-value estimate of the state action pair $(s,a)$ at time $t$. Then we have:

$$U_{t+1}(s,a) = 1/m \sum_{i=1}^{m} U_{k_i}(s_{k_i}) + \eps_1 \geq 1/m \sum_{i=1}^{m} U^*(s_{k_i}) + \eps_1 \geq U^*(s,a).$$

The first step of the inequality follows by construction of the Algorithm, the second inequality follows from the induction hypothesis, and the last one from the equation above.
\hfill
\qed
\end{proof}

\begin{lemma}
\label{lem:fail->close}
If Assumption~\ref{ass:1} holds then: If an unsuccessful update of the estimate $U(s,a)$ occurs at time~$t$ and $\learn^U_{t+1}(s,a)=\false$ then $(s,a) \in \close_{t+1}$.
\end{lemma}
\begin{proof}
Assume an unsuccessful update of the estimate $U(s,a)$ occurs at time~$t$ and let $k_1, k_2, k_3, \ldots, k_m = t$
be the $m$ most recent visits to state $s$ while action~$a$ was played.
We consider the following possibilities: (i)~If $(s,a) \not \in \close_{k_1}$, then by Assumption~\ref{ass:1} the attempt to update the U-value estimate $U(s,a)$ at time $t$ will be 
successful and there is nothing to prove. (ii)~Assume $(s,a) \in \close_{k_1}$ and there exists $i \in \set{2,m}$ such that $(s,a) \not \in \close_{k_i}$. It follows there must have been
a successful update of the U-value estimate between times $k_1$ and $k_m$, therefore the learning flag $\learn^U(s,a)$ will be set to $\true$, and there is nothing to prove.
(iii)~For the last case we have for all $i \in \set{1,m}$ that $(s,a) \in \close_{k_i}$, in particular $(s,a) \in \close_{k_m} = \close_{t}$. As the attempt
to update the U-value estimate at time $t$ was not successful, we have that $\close_{t} = \close_{t+1}$, and therefore $(s,a) \in \close_{t+1}$. The result follows.
\hfill
 \qed
\end{proof}

\begin{lemma}
\label{lem:outsideClose}
The number of timesteps $t$ such that a state-action pair $(s,a)\not  \in \close_t$  is at most $\frac{2m\vert\states\vert \vert\act\vert }{\eps_1}$.
\end{lemma}
\begin{proof}
We show that whenever $(s,a) \not \in \close_t$ for some time $t$, then in at most $2m$ more visits to the state $s$ while action $a$
is played a successful update of the U-value estimate $U(s,a)$ will occur.

Assume $(s,a) \not \in \close_t$ and $\learn^U_t(s,a)=\false$. It follows that the last attempt to update the U-value estimate $U(s,a)$ was not
successful. Let $t'$ be the time of the last attempt to update $U(s,a)$. We have that $t' \leq t$ and by Lemma~\ref{lem:fail->close} we have that 
$(s,a) \in \close_{t'+1}$. It follows there was a successful update of some U-value estimate since time $t'$ and before time $t$, otherwise
$\close_{t'} = \close_t$. By the construction of the algorithm, we have that $\learn_{t+1} = \true$ and by Assumption~\ref{ass:1} the next
attempt to update the U-value estimate $U(s,a)$ will be successful.

Assume $(s,a) \not \in \close_t$ and $\learn_t(s,a)=\true$. It follows from the construction of the algorithm, that in at most $m$ more
visits to state $s$ while action $a$ is played an attempt to update the estimate $U(s,a)$ will occur. Suppose this attempt
takes place at time $q \geq t$ and the $m$ most recent visits to state $s$ while action $a$ was played  happened at times
$k_1, k_2, \ldots, k_m = q$. There are two possibilities: (i)~If $(s,a) \not \in \close_{k_1}$ then by Assumption~\ref{ass:1} 
the attempt to update the estimate $U(s,a)$ at time $q$ will be successful; (ii)~If $(s,a) \in \close_{k_1}$, we have that $\close_{k_1} \not = \close_{t}$.
It follows there was a successful update of some U-value estimate ensuring that the learning flag $\learn_t(s,a)$ remains set to true even if the
update attempt at time $q$ will not be successful. If the update at time $q$ is not successful, it follows that $(s,a) \not \in \close_{q+1}$, and by Assumption~\ref{ass:1}
the next attempt to update $U(s,a)$ will succeed.

By Lemma~\ref{lem:suc_updates} the number of successful updates of the U-value estimate $U(s,a)$ is bounded by $\frac{\vert\states\vert \vert\act\vert }{\eps_1}$ and by the previous arguments we have
that whenever for some $t$ we have that $(s,a) \not \in \close_t$ then in at most $2m$ more visits to state $s$ while action $a$ is played, there
will be a successful attempt to update the estimate $U(s,a)$. The desired result follows.
\hfill
 \qed
\end{proof}

\begin{lemma}
\label{lem:finHor}
Let $\mdp$ be an MDP, $K$ a set of state action pairs, $\mdp'$ an arbitrary MDP, 
that coincides with $M$ on $K$ (identical transition function and identical (non-)target state),
$\straa$ a strategy, and $T \in \Nset$ a positive number. Let $A_M$ be the event that a state-action pair not in $K$ is 
encountered in a trial generated from the state $s_0$ in $\mdp$ by following strategy $\straa$ for $T$ turns. Then,

$$U^{\straa}_{M}(s_0,T) \geq U^{\straa}_{M'}(s_0,T) - \Prb(A_M)$$
\end{lemma}

\begin{proof}
For a finite path $p_T = s_0 a_0 s_1 a_1 \cdots a_{T-1} s_T \in \fpaths_s$, let $\Prb^\straa_M(p_T)$ denote the probability of path $p_T$ in $\mdp$ when
executing strategy $\straa$ from state $s_0$. Let $K_T$ be the set of all paths $p_T$ starting in $s_0$ of length $T$ such that all the state-action
pairs $(s,a)$ in $p_T$ are in $K$.
  Similarly, we will write $\neg K_T$ for
all paths of length $T$ starting in $s_0$ not in $K_T$, i.e., there exists a state-action pair $(s,a)$ in the $p_T$ such that $(s,a) \not \in K$.
Let $\reached(p_T)$ be 
a function that returns $1$ if the path $p_T$ reaches the target state $\one$, and $0$ otherwise.
Now we have the following:

\begin{align*}
U^{\straa}_{M'}(s_0,T)-U^{\straa}_{M}(s_0,T) &=  \sum_{p_T \in K_T} (\Prb^\straa_{M'}(p_T) \reached(p_T) - \Prb^\straa_M(p_T) \reached(p_T)) +\\
& \sum_{p_T \in \neg K_T} (\Prb^\straa_{M'}(p_T) \reached(p_T) - \Prb^\straa_M(p_T) \reached(p_T))\\
& = \sum_{p_T \in \neg K_T} (\Prb^\straa_{M'}(p_T) \reached(p_T) - \Prb^\straa_M(p_T) \reached(p_T)) \\
& \leq \sum_{p_T  \in \neg K_T} (\Prb^\straa_{M'}(p_T) \reached(p_T)) \\
& \leq \sum_{p_T \in \neg K_T} (\Prb^\straa_{M'}(p_T)) = \Prb(A_M)
\end{align*}
The first step in the derivation above splits the sum, according to the set $K_T$. The first term can be eliminated as for paths $p_T$ in $K_T$ visit only
states-action pairs that are common to both MDPs.

\hfill
 \qed

\end{proof}

\begin{lemma}
\label{lem:apBound}
Given a Markov chain $M$, a state $s$ in the Markov chain, $p_m$ the minimal positive transition probability, and $\tau \in \reals^+$, then for $T \geq \frac{\states \cdot \ln(2/\tau)}{{p_m}^{\states}}$ we have:
$$ V_M(s) - V_M(s,T) \leq \tau$$
\end{lemma}

\begin{proof}
We can express $V_M(s)$ as a sum of $V^{\leq T}_M(s)$ the probability to reach the target state within $T$ timesteps and $V^{>T}_M(s)$ the probability 
to reach the target state for the first time after $T$ steps. Then:
$$ V_M(s) - V_M(s,T)   = V^{\leq T}_M(s) + V^{> T}_M(s) - V_M(s,T) = V^{> T}_M(s) $$
It follows we need to show that $V^{> T}_M(s) \leq \tau$. We have by Lemma~23 of~\cite{BKK11}, that 
$ V^{> T}_M(s) \leq 2c^T$, where $c=  e^{-x^n/n}$.

\begin{align*}
2c^T & \leq \tau \hspace{6em} \Leftrightarrow \\
T \ln c & \geq \ln \tau/2  \hspace{4em} \Leftrightarrow  \\
T & \geq \frac{\ln \tau/2}{\ln c} \hspace{4em} \Leftrightarrow  \\
T & \geq \frac{\ln \tau/2}{-\frac{p_m^\states}{\states}} \hspace{4em} \Leftrightarrow  \\
T & \geq  \frac{\states \ln 2/\tau}{p_m^\states} \hspace{4em} 
\end{align*}
\hfill
\qed
\end{proof}

\begin{lemma}
\label{lem:upper}
Let $\mdp$ be an MDP that satisfies Assumption~\ref{ass:correctMDP} and $\eps, \delta \in \reals^{+}$ two positive real numbers. If Algorithm~\ref{alg:pac} is executed on
$\mdp$ it follows an $\eps/2$-optimal strategy on all but $\mathcal{O}(\frac{\zeta T}{\eps_2}\ln(\frac{1}{\delta}))$ steps with probability at least  $1-\delta/2$.
\end{lemma}
\begin{proof}

Suppose Algorithm~\ref{alg:pac} is run on a MDP $\mdp$. We assume Assumption~\ref{ass:1} holds, and that $U_t(s,a) \geq U^*(s,a)$ holds
for all time-steps and all state-action pairs $(s,a) \in \states \times \act$. By Lemmas~\ref{lem:a1Holds} and \ref{lem:uppBound} we have
that the probability that either one of these assumption is broken is at most $\frac{2\delta}{6}$.

Consider a timestep $t$, and let $\pacStraa_{t}$ denote the strategy executed by Algorithm~\ref{alg:pac}. Let $\greStraa_t$ be the memoryless
strategy given by the U-value estimates at time $t$, i.e., $\greStraa_t(s) = \argmax_{a \in \enab(s)} U_t(s,a)$. Let $s_t$ be the state of the
MDP occupied at time $t$.

We define a new MDP $\mdp'$, that is identical to the original MDP $\mdp$ on state-action pairs that are in $\close_t$. Let $\one$ (resp. $\zero$) be the target 
(resp. losing absorbing) state in MDP $\mdp$. Given a state-action pair $(s,a) \not \in \close_t$, we define the probability
to reach the target state $\one$ from $s$ while playing $a$ to $U_t(s,a)$ and with the remaining probability $1-U_t(s,a)$ the loosing absorbing
state $\zero$ is reached. 

Let $T \geq \frac{\states \ln 2 / \eps_2}{{p_m}^{\states}}$ (see Lemma~\ref{lem:apBound}), such that 
$V^{\greStraa_t}_{\mdp'}(s_t) - V^{\greStraa_t}_{\mdp'}(s_t,T)  \leq \eps_2$.
Let $\Prb(A_M)$ denote the probability of reaching a state-action pair $(s,a)$ not in $K_t$, while playing the strategy $\pacStraa_t$ from state $s_t$
in MDP $\mdp$ for $T$ turns. Let $\Prb(U)$ denote the probability of performing a successful update of the U-value estimate of some
state-action pair $(s,a)$, while playing the strategy $\pacStraa_t$ from state $s_t$ in MDP $\mdp$ for $T$ turns.
We have that:
\begin{align*}
V^{\pacStraa_t}_M(s_t,T) &\geq V^{\pacStraa_t}_{M'}(s_t,T)  - \Prb(A_M)\\
& \geq V^{\greStraa_t}_{M'}(s_t,T)  - \Prb(A_M) - \Prb(U) \\
& \geq V^{\greStraa_t}_{M'}(s_t) - \eps_2  - \Prb(A_M) - \Prb(U) 
\end{align*}
The first step follows from Lemma~\ref{lem:finHor}, the second inequality follows from the fact that $\pacStraa_t$ behaves as $\greStraa_t$ as long as no
U-value estimate is changed. The last step follows from Lemma~\ref{lem:uppBound}.

Next we consider two mutually exclusive cases. 

\smallskip\noindent\textbf{First case:} First suppose that $\Prb(A_M) + \Prb(U) \geq \eps_2$, i.e.,
by following strategy $\pacStraa_t$ the algorithm will either perform a U-value estimate update in $T$ timesteps  or encounter a
state action pair $(s,a) \not \in \close_t$ in $T$ timesteps, with probability at least $\eps_2 / 2$ (since $\Prb(A_M \text{ or } U) \geq (\Prb(A_M) + \Prb(U))/2$).
By Lemma~\ref{lem:suc_updates} the former event cannot happen more that $\frac{\vert\states\vert \vert\act\vert}{\eps_1}$ times and by Lemma~\ref{lem:outsideClose} the latter
event cannot happen more than $\frac{2m\vert\states\vert \vert\act\vert }{\eps_1}$ times. We are interested in the number of steps
after which every state-action pair $(s,a)$ will have its U-value estimate updated $\frac{1}{\eps_1}$ times with probability at
least $1 - \frac{\delta}{6}$.

Let $\zeta =(2m+1)\frac{\vert\states\vert \vert\act\vert}{\eps_1}$, $u = 4 \ln (\frac{6}{\delta}) +1$, and $\gamma = \frac{u-1}{u}$. We also assume that the probability of event $A_M \text{ or } U$
happening in $T$ timesteps is exactly $\eps_2$, as higher probabilities can only decrease the number of steps needed for updating all of the U-value
estimates with sufficiently high probability. 

Let $k = \frac{\zeta}{\eps_2} u$. We define a random variable $X_i$ for $0 \leq i \leq k$
 that is equal to $1$ if event $A_M \text{ or } U$ happened 
between times $i T$ and $(i+1) T $ and $0$ otherwise, and let $S = \sum_{i=0}^{k} X_i$.
We want to show that $\Prb(S \leq \zeta) \leq \frac{\delta}{6}$.

By a variant of Chernoff bound~\cite{AV79} and the fact that $\expect[S] \geq k \cdot \eps_2$ we have: 
$$\Prb(S < (1-\gamma) k \cdot \eps_2)  \leq \Prb(S < (1-\gamma) \expect[S]) \leq e^{\frac{-\gamma^2 \expect[S]}{2}}  \leq e^{\frac{-\gamma^2 k \cdot \eps_2}{2}}$$
As $(1-\gamma) k \cdot \eps_2 = (1-\gamma) \zeta \cdot u = \zeta$, we have that $\Prb(S \leq \zeta) \leq e^{\frac{-\gamma^2 \zeta u}{2}}$ and it remains to show that
  $e^{\frac{-\gamma^2 \zeta u}{2}} \leq  \frac{\delta}{6}$:

  $$e^{\frac{-\gamma^2 \zeta u}{2}} = e^{\frac{-\left(\frac{(u-1)}{u}\right)^2 \zeta u}  {2}}  =  e^{-\frac{(u-1)^2}{2 u} \zeta} \leq e^{-\frac{(u-1)^2}{2 u} } $$
  $$= e^{-\frac{2(u-1)(u-1)}{4 u}} \leq e^{-\frac{u-1}{4}}  = e^{-\frac{(4 \ln (\frac{6}{\delta}) +1)-1}{4} }=  e^{-\ln((\frac{6}{\delta}))} = \frac{\delta}{6}$$

It follows that after $\mathcal{O}(\frac{\zeta T}{\eps_2}\ln(\frac{1}{\delta}))$ timesteps where $\Prb(A_M) + \Prb(U) \geq \eps_2$ all the U-value estimates are updated $\frac{1}{\eps_1}$ times with probability at least $1 - \frac{\delta}{6}$, and by Lemma~\ref{lem:suc_updates} no further updates are possible.

\smallskip\noindent\textbf{Second case:}
For the second case suppose that $\Prb(A_M) + \Prb(U) < \eps_2$, we prove the following statement for all states $s$ in Lemma~\ref{lem:matrix}
$$0 \leq U_t(s) - V^{\greStraa_t}_{M'}(s) \leq 3 \cons \eps_1$$
It follows that:
\begin{align*}
V^{\pacStraa_t}_M(s_t) & \geq V^{\pacStraa_t}_M(s_t,T) \\
& \geq V^{\greStraa_t}_{M'}(s_t) - \eps_2  - \Prb(A_M) - \Prb(U) \\
& \geq V^{\greStraa_t}_{M'}(s_t) - \eps_2  - \eps_2\\
& \geq U_{t}(s_t) - 3 \cons \eps_1 - 2 \eps_2  \\
& \geq V^*(s_t) - 3\cons \eps_1 - 2 \eps_2  \\
\end{align*}
By setting $\eps_1 = \frac{\eps}{4\cons} = \frac{\eps \cdot (p_{\min} / \max_{s \in \states} \enab(s))^{\vert\states\vert}}{12 \vert\states\vert}$  and $\eps_2 = \eps/8$ we get the desired results, i.e., 
$$V^{\pacStraa_t}_M(s_t) \geq V^*(s_t) - \eps/2$$
is true on all but $\mathcal{O}(\frac{\zeta T}{\eps_2}\ln(\frac{1}{\delta}))$ timesteps.

\hfill
\qed

\end{proof}

\begin{lemma}
\label{lem:matrix}
Let $s\in \states$ be a state and $V_t$, $V^{\greStraa_t}_{M'}$ are defined as in the proof of Lemma~\ref{lem:upper}, and $\cons = \frac{\vert \states \vert}{{(p_{\min} / \max_{s \in \states}\enab(s))}^{\vert \states \vert}}$ then for all states $s\in \states$ we have:
$$  0 \leq U_t(s) - V^{\greStraa_t}_{M'}(s) \leq 3 \cons  \eps_1$$
\end{lemma}
\begin{proof}
Note that $V^{\greStraa_t}_{M'}$ is the least fixpoint of the following set of Bellman equations:
\begin{align*}
V^{\greStraa_t}_{M'}(\one) & =1\\
V^{\greStraa_t}_{M'}(\zero) & =0\\
V^{\greStraa_t}_{M'}(s) & = \sum_{s' \in \states, a \in \support{\greStraa_t(s)}} \greStraa_t(s)(a)\trans(s,a)(s') \cdot V^{\greStraa_t}_{M'}(s') \hspace{3em} \text{ for } (s,\greStraa_t(s)) \in \close_t \\
V^{\greStraa_t}_{M'}(s) & = U_t(s,\greStraa_t(s)) \hspace{17.1em} \text{ for } (s,\greStraa_t(s)) \not \in \close_t \\
\end{align*}

As the input MDP $\mdp$ satisfies Assumption~\ref{ass:correctMDP}, it follows that also the modified MDP $\mdp'$ satisfies Assumption~\ref{ass:correctMDP}, as no new ECs are introduced.
One can show that whenever an MDP satisfies Assumption~\ref{ass:correctMDP} there exists a unique fixpoint of the Bellman equations above.

Note that $\greStraa_t(s)$ plays uniformly at random actions $a$  that maximize $U_t(s,a)$. Similarly $U_t$ is the greatest fixpoint of the following set of equations:
\begin{align*}
U_t(\one) & =1\\
U_t(\zero) & =0\\
U_t(s) & = \max_{a\in \act} Q_t(s,a) = \sum_{a \in \support{\greStraa_t(s)}}U_t(s,a) \\
& \leq \sum_{s' \in \states, a \in \support{\greStraa_t(s)}} \greStraa_t(s)(a)\trans(s,a)(s') \cdot U_t(s')  + 3 \eps_1 \hspace{3em} \text{ for } (s,\greStraa_t(s)) \in \close_t\\
U_t(s) & = U_t(s,\greStraa_t(s)) \hspace{19.0em} \text{ for } (s,\greStraa_t(s)) \not \in \close_t \\
\end{align*}
where every inequality  given a fixed $\greStraa_t$ can be viewed as a equality $U_t(s)=\sum_{s' \in \states, a \in \support{\greStraa_t(s)}} \greStraa_t(s)(a)\trans(s,a)(s') \cdot U_t(s')  + c^{\greStraa_t}_s$ for some positive $c^{\greStraa_t}_s$ bounded by $3 \eps_1$.
It follows from Assumption~\ref{ass:correctMDP} and Lemma~\ref{lem:perturbation} that also the equations for $U_t$ have a unique fixpoint. We need to bound for all states
$s \in \states$ the difference between $U_t(s)$  $V^{\greStraa_t}_{M'}(s)$ in terms of $\eps_1$.

One can also view the equations for $U_t$ as assigning a positive cost bounded by $3 \eps_1$ to every move of the strategy before the terminal state $\zero$ or $\one$ are reached.
These two states are reached in the Markov chain obtained by playing strategy $\greStraa_t$ with probability $1$. This follows from Assumption~\ref{ass:correctMDP} and from the fact the strategy $\greStraa_t$ plays uniformly all the actions that maximize $U_t(s)$. Every EC in $\mdp'$ satisfies that all the states in the EC (except state $\zero$) have value $1$. It follows that from
every EC in $\mdp'$ with the exception of the terminal states $\one$ and $\zero$ (i)~there exists an action that with positive probability leaves the EC, and (ii)~this action is played by $\greStraa_t$ with positive
probability.

We denote by $c_{\min}$ the lower bound on the minimal transition probability in the Markov chain is $p_{\min}/\max_{s \in \states}\enab(s)$. The probability to
reach the terminal states $\zero, \one$ in $\vert S \vert$ steps is bounded from below by $c_{\min}^{\vert \states \vert}$. The probability not to reach the terminal states
in $\vert \states \vert$ steps is therefore $1- c_{\min}^{\vert \states \vert}$. The expected cost to reach the terminal states is bounded by:
$$3 \eps_1  \vert \states \vert  \sum_{n=0}^{\infty} (1-c_{\min}^{\vert S \vert})^n  = 3 \eps_1  \vert \states \vert  \frac{1}{ c_{\min}^{\vert \states \vert}} =  3 \eps_1  \vert \states \vert  \frac{1}{{(p_{\min} / \max_{s \in \states}\enab(s))}^{\vert \states \vert}}$$

\end{proof}

\smallskip\noindent\textbf{Discussion about the lower bound estimates.}
The case for the lower bounds is simpler, as at timestep $t$ the current greedy strategy $\greStraa^t$  is not influenced by the value 
of the lower bound estimates $L(s,a)$. By dual arguments to the case 
of upper bounds, one can show, that with high probability the lower bound estimates $L(s,a)$ are actual lower bounds. By Lemma~\ref{lem:upper} we have, 
that after $\mathcal{O}(\frac{\zeta T}{\eps_2}\ln(\frac{1}{\delta}))$
steps, the memoryless strategy $\greStraa^*$ determined by the upper bounds is $\eps/2$ optimal with probability $1 - \delta/2$ and no further improvement of the strategy $\greStraa^*$ will occur.
Once we fix the strategy $\greStraa^*$ and the MDP $\mdp$ we obtain a Markov chain in which the lower bounds are being propagated for $\mathcal{O}(\frac{\zeta T}{\eps_2}\ln(\frac{1}{\delta}))$ timesteps, in order to increase the estimates of the lower bound $\eps$ close to the actual value. This fact together with Lemma~\ref{lem:upper} establishes the main theorem:

\begin{theorem}\label{thm:PAC-1MEC}
Let $\mdp$ be an MDP that satisfies Assumption~\ref{ass:correctMDP} and $\eps, \delta \in \reals^{+}$ two positive real numbers. If Algorithm~\ref{alg:pac} is executed on
$\mdp$ it follows an $\eps/2$-optimal strategy on all but $\mathcal{O}(\frac{\zeta T}{\eps_2}\ln(\frac{1}{\delta}))$ steps and $U(s_0) - L(s_0) \leq \eps$ with probability at least  $1-\delta$.
\end{theorem}

\section{Proofs of Lemmata \ref{lemma:mec-explore}, \ref{lem:proc} and \ref{lem:mec-is-ec}: Correctness of EC Identification and Collapsing}

\subsection{Proofs of Lemma \ref{lemma:mec-explore}}

\begin{reflemma}{lemma:mec-explore}
Let $c=\exp\left(-\left(p_{\min}/\maxE\right)^{\kappa}\,/\,\kappa\right)$ where $\kappa=|S|\maxE+1$. 
Let $i\geq \kappa$. Assume that the \textsc{Explore} phase in Algorithm~\ref{alg:skeleton} terminates with probability less than $1$. Then, provided the \textsc{Explore} phase does not terminate within $3 i^3$ iterations, the conditional probability that $\app{i}$ is an EC is at least $1-2c^i i^3\cdot \left(p_{\min}/\maxE\right)^{-\kappa}$.
\end{reflemma}
\begin{proof}[Sketch]
The main idea behind the proof is following. Each execution of the explore phase simulates a path $\omega$ of $\mdp$ according to the memoryless strategy determined by the function $U$. In fact, $\omega$ can be seen as a path in a Markov chain $\mathit{MC}$ (i.e., a MDP where every state has exactly one enabled action) obtained from $\mdp$ by fixing the memoryless strategy. Here states of $\mathit{MC}$ correspond to state-action pairs $(s,a)$ of $\mdp$ such that $a$ is chosen in $s$ with a positive probability (we also add an initial state $\sinit$ where the first action is chosen). The chain $\mathit{MC}$ is constructed in such a way that each bottom scc~\footnote{A bottom scc (bottom strongly connected component) is a maximal set $D$ of states (with respect to the subset ordering) such that for all states $s,s'\in D$ the state $s'$ is reachable from $s$ with a positive probability and no state outside of $D$ is reachable from $s$.} corresponds to an end-component in $\mdp$.

Applying Lemma~23 of~\cite{BKK11}, we obtain a bound on the probability that a path starting in $\sinit$ visits a bottom scc of $\mathit{MC}$ in at most $i$ steps. 
Using the same lemma we also bound the probability that a path of $\mathit{MC}$ starting in a state of a bottom scc visits all states of this bottom scc $i+1$ times within $\kappa i (i+1)$ steps.
Putting these two bounds together, we obtain that with probability at most $1-2c^i i^3$, the first $2 i^4$ steps of a path starting in $\sinit$ visit all states of the bottom scc $i+1$ times and all other states at most $i$ times. 
Observe that this bottom scc may contain $\one$, or $\zero$ in which case the  \textsc{Explore} phase terminates within $2i^4$ iterations.
It follows, that with probability at least $1-2c^i i^3$, the \textsc{Explore} phase in Algorithm~\ref{alg:skeleton} either terminates within $2 i^4$ iterations, or $\app{i}$ is an EC.

Finally, to obtain the conditional probability, we observe that if there is a
 bottom scc in $\mathit{MC}$ reachable from $\sinit$ that does not contain
  $\one$, or $\zero$, then such bottom scc is reachable with probability at least
   $\left(p_{\min}/\maxE\right)^{\kappa}$. Using simple probability theory and
    algebra, we obtain that the desired conditional probability is
     $1-2c^i i^3\cdot \left(p_{\min}/\maxE\right)^{-\kappa}$.
\qed
\end{proof}
\begin{proof}
In what follows we denote by $\maxU(s)$ the set of all actions $a\in \enab(s)$ that maximise $U(s,a)$.
Note that the \textsc{Explore} phase samples an infinite path of a finite-state Markov chain $\mathit{MC}$ (i.e., a MDP where every state has exactly one enabled action) whose set of states is $\{\sinit\}\cup \{(s,a)\mid s\in S\wedge a\in \enab(s)\}$ and transitions are defined as follows: 
There is a transition with probability $x$ from $\sinit$ to $(\sinit,a)$ iff $a\in \maxU(\sinit)$ and  $x=1\,/\,|\maxU(\sinit)|$. There is a transition with probability $x$ from $(s,a)$ to $(s',a')$ iff $a'\in \maxU(s')$ and $x=\Delta(s,a)(s')/|\maxU(s')|$
(In other words, $x$ is the probability that $s'$ follows $(s,a)$ and then $a'$ is chosen by the \textsc{Explore} phase in $s'$.) 
Note that in $\exp\left(-\left(p_{\min}/\maxE\right)^{\kappa}\,/\,\kappa\right)$, the number $p_{\min}/\maxE$ is less than or equal to the minimum positive transition probability in $\mathit{MC}$, and $\kappa$ is the number of states of $\mathit{MC}$.

A bottom scc (or a recurrent class) of a Markov chain is a maximal set (with respect to inclusion) of states $D$ such that for all $s,s'\in D$ the state $s'$ is reachable from $s$ with positive probability and no state outside of $D$ is reachable from any state of $D$ with positive probability. It is well known that almost every infinite path initiated in any state of a finite $\mathit{MC}$ visits all states of some bottom scc infinitely many times. Also, observe that each bottom scc of $\mathit{MC}$ determines an end-component of $\mdp$ in a natural way.

Assume that there is a bottom scc reachable from $\sinit$ that contains neither $\one$, nor $\zero$. Let us denote by $R_{app}$ the 
set of all infinite paths of $\mathit{MC}$ starting in $\sinit$ that within $2i^3$ steps
\begin{itemize}
\item visit all states of a bottom scc at least $i+1$ times,
\item visit all other states at most $i$ times.
\end{itemize} 
Let $R_{\one,\zero}$ be the set of all infinite paths that visit $\{\one,\zero\}$ within $2i^3$ steps and let $\bar{R}_{\one,\zero}$ be the complement of $R_{\one,\zero}$. Let $P_{\one,\zero}$ be the probability $\Prb_{\mathit{MC},\sinit}\left(R_{\one,\zero}\right)$, and let $\bar{P}_{\one,\zero}=1-P_{\one,\zero}=\Prb_{\mathit{MC},\sinit}\left(\bar{R}_{\one,\zero}\right)$. 
As bottom scc of $\mathit{MC}$ determine end-components of $\mdp$ we obtain the following: Assuming that the \textsc{Explore} phase in Algorithm~\ref{alg:skeleton} {\em does not} terminate within $2 i^3$ iterations, $\app{i}$ is an end-component with (conditional) probability at least $\Prb_{\mathit{MC},\sinit}\left(R_{app}\mid \bar{R}_{\one,\zero}\right)$. So it suffices to bound 
the conditional probability $\Prb_{\mathit{MC},\sinit}\left(R_{app}\mid \bar{R}_{\one,\zero}\right)$.

First, we show the following inequality:
\begin{equation}\label{eq:appear}
\Prb_{\mathit{MC},\sinit}\left(R_{app}\right)\quad \geq \quad 1-2c^i i^3 
\end{equation}
By Lemma~23 of~\cite{BKK11}, with probability at most $2c^i$, an infinite path of $\mathit{MC}$ starting in $\sinit$ {\em does not} visit a bottom scc
of $\mathit{MC}$ within $i$ steps. 
Let $D=\{s_0,s_1,\ldots,s_k\}$ be a bottom scc of $\mathit{MC}$. Observe that with probability at most $2c^i$, an infinite path starting in $s_{\ell}$ {\em does not} visit $s_{\ell+1\mod k}$ within $i$ steps. It follows that with probability at most $k 2 c^i$, an infinite path starting in $s_{\ell}$ {\em does not} visit all states of $D$ within $k i$ steps. However, then with probability at most $(i+1) k 2 c^i\leq (i+1)\kappa 2 c^i$, 
an infinite path starting in $s_{\ell}$ {\em does not} visit all states of $D$ at least $i+1$ times within $k i (i+1)\leq \kappa i (i+1)$ steps. 
Finally, with probability at most $2c^i+(i+1) \kappa 2 c^i=2c^i (1+i\kappa+\kappa)$, an infinite path starting in $\sinit$ either fails to reach a bottom scc within $i$ steps, or reaches a bottom scc within $i$ steps but fails to subsequently reach all states of this bottom scc at least $i+1$ times within $\kappa i (i+1)$ steps.
Thus, with probability at least $1-2c^i (1+i\kappa+\kappa)\geq 1-2c^i i^3$, an infinite path starting in $\sinit$ visits an end-component within $i$ steps and then all states of this bottom scc at least $i+1$ times within $\kappa i (i+1)\leq 2i^3$ steps. This proves Equation~(\ref{eq:appear}).

Now it is easy to see that
\begin{equation}\label{eq:term}
\bar{P}_{\one,\zero}\quad \geq \quad \left(p_{\min}/\maxE\right)^{\kappa}
\end{equation}
Then the desired conditional probability satisfies:
\begin{align*}
\Prb_{\mathit{MC},\sinit}\left(R_{app}\mid \bar{R}_{\one,\zero}\right) & =  
	\Prb_{\mathit{MC},\sinit}\left(R_{app}\cap \bar{R}_{\one,\zero}\right)\,/\,\bar{P}_{\one,\zero}  \\
	& =  \Prb_{\mathit{MC},\sinit}\left(R_{app}\smallsetminus R_{\one,\zero}\right)\,/\, \bar{P}_{\one,\zero}\\
	& =  \left(\Prb_{\mathit{MC},\sinit}\left(R_{app}\right)-\Prb_{\mathit{MC},\sinit}\left(R_{\one,\zero}\cap R_{app}\right)\right)\,/\,\bar{P}_{\one,\zero} \\
	& \geq  \left(\Prb_{\mathit{MC},\sinit}\left(R_{app}\right)-P_{\one,\zero}\right)\,/\,\bar{P}_{\one,\zero} && \text{ by (\ref{eq:appear})}\\
	& \geq  \left(1-2c^i i^3-\left(1-\bar{P}_{\one,\zero}\right)\right)\,/\,\bar{P}_{\one,\zero} \\
	& =  \left(\bar{P}_{\one,\zero}-2c^i i^3\right)\,/\,\bar{P}_{\one,\zero} \\
	& =  1-\left(2c^i i^3\,/\,\bar{P}_{\one,\zero}\right) \\
	& \geq  1-2c^i i^3\cdot \left(p_{\min}/\maxE\right)^{-\kappa} && \text{ by (\ref{eq:term})}
\end{align*}
\qed\end{proof}

\subsection{Proofs of Lemma \ref{lem:proc}}

\begin{reflemma}{lem:proc}
Assume $(\xu,\xg)$ is an EC in MDP $\mdp$,  $V_\mdp$ the value before the \textsc{Process ECs} procedure in Algorithm~\ref{alg:skel-mec}, and $V_{\mdp'}$ the value after the procedure, then:
\begin{compactenum}
\item For $i \in \{0,1\}$ if $\makeTerminal(s_{(\xu,\xg)},i)$  is called, then $ \forall s\in \xu: \ V_{\mdp}(s)=i$;
\item $\forall s \in \states \setminus \xu: \  V_\mdp(s) = V_{\mdp'}(s);$
\item $\forall s \in \xu:  \  V_\mdp(s) = V_{\mdp'}(s_{(\xu,\xg)});$
\end{compactenum}
\end{reflemma}
\begin{proof}

\textbf{Point 1.} The function $\makeTerminal(s_{(\xu,\xg)},\zero)$ is called in Algorithm~\ref{alg:skel-mec} only if there are no actions
 available in state $s_{(\xu,\xg)}$ and $\xu \cap F = \emptyset$. It follows that the support of all the actions that were enabled in
  states in $\xu$ in MDP $\mdp$ stays in $\xu$, i.e, there is no action leaving the set $\xu$ . As $\xu \cap F = \emptyset$, it
   follows that for all states in $\xu$ the probability to reach the target state is $0$ and therefore 
   $ \forall s\in \xu: \ V_{\mdp}(s)=0$ .

The function $\makeTerminal(s_{(\xu,\xg)},\one)$ is called in Algorithm~\ref{alg:skel-mec} only if $\xu \cap F \neq \emptyset$. A strategy
in MDP $\mdp$ that plays in state $s \in \xu$ all the actions $\enab(s) \cap \xg$ uniformly at random, will visit all the states in $\xu$ almost
surely. It follows that from every state $s \in \xu$ the target set is reached almost surely. It follows that $ \forall s\in \xu: \ V_{\mdp}(s)=1$

\textbf{Points 2 and 3.}
These two points follow directly from Theorem~2 of~\cite{CBGK08}.
\qed\end{proof}

\subsection{Proofs of Lemma \ref{lem:mec-is-ec}}

\begin{reflemma}{lem:mec-is-ec}
Let $\mdp$ and $\mdp^T$ be the MDPs from the construction above and $T$ the corresponding set of explored states.
 Then every MEC $(\xu,\xg)$ in $\mdp^T$ such that $\xu \subseteq T$ is an EC in $\mdp$.
\end{reflemma}
\begin{proof}
Let $\mdp=\langle \states,\sinit,\act, \enab, \Delta\rangle$ and $\mdp^T=\langle T',\sinit,\act', \enab', \Delta'\rangle$ be the two
MDPs, and let $(\xu,\xg)$ be a MEC in $\mdp^T$ such that $\xu \subseteq T$. As $T \subseteq \states$ we have
that the states of $\xu$ are present in MDP $\mdp$. The three other required properties

\begin{enumerate}
\item $\xg \subseteq \bigcup_{s \in \xu} \enab(s)$;
\item if $s \in \xu$, $a \in \xg$, and $\Delta(s,a)(s')> 0 $ then $s' \in \xu$; and
\item for all $s,s' \in \xu$, there exists a path $\omega = s_0 a_0 s_1 a_1 \ldots s_n$ such that $s_0 =s$, $s_n = s'$, and for all
$0 \leq i < n$ we have that $a_i \in \xg$ and $\Delta(s_i,a)(s_{i+1})>0$;
\end{enumerate}

follow easily from the fact that for all states $s \in \xu$ and actions $a \in \xg$ we have: $\enab(s) = \enab'(s)$; $\Delta(s,a) = \Delta'(s,a$); and 
 $(\xu,\xg)$ is an EC in $\mdp^T$.
\qed
\end{proof}

\section{Proof of Theorem~\ref{thm:BRTDP-mec}: Correctness of OBRTDP}
Consider Algorithm~\ref{alg:skeleton} with line 17 being ``\textbf{until} false'', i.e. iterating the outer repeat loop ad infinitum (we prove that in this situation, the original expression $\max_a U(\sinit,a)-\max_a L(\sinit,a)$ from line 17  goes to zero).
As the MDP $\mdp$ may change during computation of the learning algorithm, we denote by $\mdp_i=\langle \states_i,\xi_i,\act_i, \enab_i, \Delta_i\rangle$ the current MDP after $i$ iterations of the outer repeat-until cycle of Algorithm~\ref{alg:skeleton}. 
Each $\mdp_i$ is obtained from $\mdp$ by possibly several collapses of end-components.
Recall that in an MDP $\mdp'$ obtained by collapsing $(\xu,\xg)$, the state $s_{(\xu,\xg)}$ corresponds to the set of states $R$ and, in particular, $V_{\mdp}(s_{(\xu,\xg)},a)=V_{\mdp'}(s,a)$ for all actions $a$ that are enabled both in $s_{(\xu,\xg)}$ and in $s$.
Thus slightly abusing notation we may consider states of each $\mdp_i$ to be sets of states of the original MDP $\mdp$. 
So given a state $\xi\in \states_i$ of $\mdp_i$, we write $s\in \xi$ to say that the state $s\in S$ of $\mdp$ belongs to (or corresponds to) the state $\xi$. 

Note that  $V_{\mdp}(s,a)=V_{\mdp_i}(\xi,a)$ for $s\in \xi\in \states_i$ and all $a\in \enab_i(\xi)$.
Thus, in what follows we use $V(s,a)$ to denote $V_{\mdp}(s,a)$. 
We also denote by $U_i$ and $L_i$ the functions $U$ and $L$ after $i$ iterations.
Observe that $U_i, L_i:\states_i\times \act_i\rightarrow [0,1]$. We extend $U_i$ and $L_i$ to states of $S$ by $U_i(s,a):=U_i(\xi,a)$ and $L_i(s,a):=L_i(\xi,a)$ for $s\in \xi\in \states_i$ and all $a\in \enab_i(\xi)$.
We also use $\enab_i(s)$ to denote $\enab_i(\xi)$ for $s\in \xi\in \states_i$.

\begin{claim}
For all $s\in S$, every $i\in\mathbb N$ and all $a\in \enab_i(s)$, 
\[U_1(s,a)\geq \cdots \geq U_i(s,a)\geq V(s,a) \geq L_i(s,a)\geq \cdots \geq L_1(s,a)\]
\end{claim}
\begin{proof}
A simple induction applies if end-components are not collapsed in the $i$-th iteration of the outer cycle of Algorithm~\ref{alg:skeleton}. Otherwise, if they are collapsed, then the claim follows from the fact that collapsing preserves the values of $U$, $L$, and $V$ (see lines 6 and 7 of Algorithm~\ref{alg:skel-mec}.
\end{proof}
It follows from the above claim that for all $a\in \bigcap_{i=1}^{\infty} \enab_i(s)$, the limits $\lim_{i\to\infty} U_i(s,a)$ and $\lim_{i\to\infty} L_i(s,a)$ are well defined and finite. As there are only finitely many actions,
$\lim_{i\to\infty}\left(\max_{a\in \enab_i(\sinit)} U_i(\sinit,a)-\max_{a\in \enab_i(\sinit)} L_i(\sinit,a)\right)$ is well defined and finite.
\begin{claim}
$\lim_{i\to\infty}\left(\max_{a\in \enab_i(\sinit)} U_i(\sinit,a)-\max_{a\in \enab_i(\sinit)} L_i(\sinit,a)\right)=0$ almost surely.
\end{claim}
\begin{proof}
Given $s\in S$, let $a_i(s)\in \enab_i(s)$ be an action maximising $U_i(s,a)$ over $\enab_i(s)$. Define $\delta_i(s):=U_i(s,a_i(s))-L_i(s,a_i(s))$.
Since $\delta_i(s)\geq \max_{a\in \enab_i(s)} U_i(s,a)-\max_{a\in \enab_i(s)} L_i(s,a)$ (expression of line 17 in the original Algorithm~\ref{alg:skeleton}),
it is sufficient to prove that $\lim_{i\to\infty}\delta_i(\sinit)=0$ almost surely.

As the function \otf{} can collapse end-components only finitely many times, every computation of the learning algorithm eventually stays with a fixed MDP $\mdp'=\langle \states',\sinit',\act', \enab', \Delta'\rangle$, i.e., almost surely $\mdp'=\mdp_k=\mdp_{k+1}=\cdots$ for some $k$. Note that $\mdp'$ is obtained by a series of collapses of end-components of $\mdp$. We call the moment from which the MDP does not change the {\em fixing point}. 

Let us denote by $\states'$ the set of states of $\mdp'$.
Note that for every $\xi\in \states'$ and for all $s,s'\in \xi$ we have $\delta(s)=\delta(s')$ since $\delta_i(s)=\delta_i(s')$ for $i$ greater than the fixing point. We denote by $\delta(\xi)$ the value $\delta(s)$ for some (all) $s\in \xi$.
Let us denote by $\act'$ the set of actions of $\mdp'$. For every $\xi\in \states'$ we denote by $\enab'(\xi)$ the set of actions enabled in the state $\xi$ of $\mdp'$. Also, the initial state, $\bar{\xi}$, of $\mdp'$ is the only state of $\mdp'$ that contains $\sinit$.

Let $\Sigma_U$ be the set of all memoryless strategies in $\mdp'$ which occur as $\sigma_{U_i}$ for infinitely many $i$ {\em after} the fixing point. Each $\sigma\in \Sigma_U$ induces a chain with reachable state space $\states'_{\sigma}$ and uses actions $\act'_{\sigma}$. Note that under $\sigma\in \Sigma_U$, all states of $\states'_{\sigma}$ will be almost surely visited infinitely often if infinitely many simulations are run. Similarly, all actions of $\act'_{\sigma}$ will be used almost surely infinitely many times. Let $\states'_{\infty}=\bigcup_{\sigma\in\Sigma_U} \states'_i$ and let $\act'_{\infty}=\bigcup_{\sigma\in \Sigma_U} \act'_{\sigma}$. During almost all computations of the learning algorithm, all states of $\states'_{\infty}$ are visited infinitely often, and all actions of $\act'_{\infty}$ are used infinitely often. 

Let $\delta=\max_{\xi\in \states'_{\infty}}\delta(\xi)$ and $D=\{\xi\in \states'_{\infty}\mid \delta(\xi)=\delta\}$. 
To obtain a contradiction, assume that $\delta>0$, which implies that $0,1\not \in D$.
We claim that $D$ cannot contain a subset $D'$ forming an end-component with any set of actions from $A'_{\infty}$. Indeed, assume the opposite is true, and $(D',G)$ is such an end component in $\mdp'$.
At least one of the strategies $\sigma\in \Sigma_U$ visits a state of $D'$ infinitely many times. As all state-action pairs $(\xi,a)\in D'\times G(\xi)$ satisfy $U_i(\xi,a)=1$ for all $i$ and $k_i\geq |S|$, almost surely a simulation of $\sigma$ of length $k_i$ visits the whole component $(D',G)$. %
 This means that \otf{} is called while $\app{0}$ contains the component $(D',G)$, which in turn means that $(D',G)$ gets collapsed, a contradiction with the assumption that the learning procedure stays fixed on $\mdp'$ after the fixing point. 

By definition of an end-component we get
$$\exists \xi\in D:\forall a\in \enab'(\xi)\cap \act'_{\infty}:supp(\Delta(\xi,a))\not\subseteq D$$
because otherwise $D$ will form a ``closed'' component with some actions of $\act'_{\infty}$, and hence would contain an end-component.
Thus for every $a\in \enab'(\xi)\cap \act'_{\infty}$ we have $\xi_a\notin D$ with $\Delta(\xi,a)(\xi_a)>0$. %
Since $\xi_a\notin D$ we have $\delta(\xi_a)<\delta$. Now for every $a\in\enab'(\xi)\cap \act'_{\infty}$ we have
\begin{align*}
\delta(\xi)&=\sum_{\xi'\in \states'_{\infty}}\Delta(\xi,a)(\xi')\cdot \delta(\xi')\\
&=\sum_{\xi'\in \states'_{\infty},\xi'\neq \xi_a}\Delta(\xi,a)(\xi')\cdot \delta(\xi)+\Delta(\xi,a)(\xi_a)\cdot \delta(\xi_a)\\
&<\sum_{\xi'\in \states'_{\infty},\xi'\neq \xi_a}\Delta(\xi,a)(\xi')\cdot \delta+\Delta(\xi,a)(\xi_a)\cdot \delta\\
&=\delta
\end{align*}
a contradiction with $\xi\in D$.

As a corollary, Algorithm~\ref{alg:skeleton} with \textsc{Update} defined in Algorithm~\ref{alg:skel-mec} and extended with calls to \otf{} almost surely terminates for any $\varepsilon>0$. Further, $U_i\geq V\geq L_i$ pointwise and invariantly for every $i$ by the first claim, the returned result is correct.
\qed
\end{proof}

\section{Proof of Theorem \ref{thm:odql}: Correctness of ODQL}

We define a sequence of random variables $(X_i)_{i=1}^\infty$ on executions of ODQL. The value of $X_i$ is $1$ if there is a call to \textsc{Collapse} taking place in the execution after the $i$-th \textsc{Explore} phase, and $0$ otherwise. 

\begin{lemma}
For any $\epsilon_3,\epsilon_4>0$, we can find $i$ such that with probability $1-\epsilon_3$ after $i$ \textsc{Explore} phases the probability that a further collapse happens is less than $\epsilon_4$. 
\end{lemma}
\begin{proof}
On each execution, the sequence $(X_i)_{i=1}^\infty$ is non-increasing. Moreover, \textsc{Collapse} can happen at most $\vert \states \vert \cdot \maxE $ times in each execution, because each invocation of \textsc{Collapse} reduces the number of states or actions. Since there are finitely many collapses on each execution, for every execution there is $i$ where $X_i$ is $0$. Thus also $\lim_{i\to\infty}\mathbb E[X_i]=0$ and we conclude by Markov inequality. 
\qed\end{proof}

We use random variable $M'$ to denote the MDP after $i$ \textsc{Explore}'s. The probability that $V$ in $\mdp$ is the same as $V$ on $M'$ (extended to $\mdp$ by ``decollapsing'') is at least $1-\epsilon_3=\epsilon_4-\epsilon_5$, where $\epsilon_5$ is the probability that at least one of the collapses merges a non-EC. By Lemma \ref{lemma:mec-explore}, we can bound the probability of erroneous collapses by the choice of $\ell_i$, we obtain $\epsilon_5<\delta/2$. %

Furthermore, when along an execution all calls to \textsc{Collapse} only collapse ECs, due to Lemma \ref{lem:proc} we can use the analysis from Theorem~\ref{thm:dql} to obtain that for any $k$, after $k$ updates the probability that the bounds $U$ and $L$ are correct is the same as in the case of the MEC-free DQL, denote it $\delta'$ and note $\delta'<\delta$, where $\delta$ is the error tolerance of DQL.%

We now show what happens with the remaining ECs:
\begin{lemma}\label{lem:EC-dichotomy}
For every $\epsilon_6,\epsilon_7>0$ there is $j$ such that with probability $1-\epsilon_6$ after $j+i$ \textsc{Explore} phases, the probability that the following holds is at least $1-\epsilon_7$: for each EC $E$ in $M'$:
\begin{enumerate}
 \item either none of the states of $E$ is ever visited after the $j+i$th \textsc{Explore},
 \item or all states of $E$ are visited infinitely often a.s. and $V(s)=1$ for each state $s$ of $E$.
\end{enumerate}
\end{lemma}
\begin{proof}
We either visit an EC only finitely often or not. In the former case, we can bound with arbitrarily high probability when the last visit happens. In the latter case, we first show we visit all states infinitely often if we visit at least one of them infinitely often.
\begin{lemma}\label{lem:EC-values}
For every EC $E$, at any moment either $U(s)=1$ for each $s\in E$, or $U(s)=0$ for each $s\in E$. 

Moreover, the latter happens only due to an invocation of \textsc{MakeTerminal} or for state $\zero$ (if there is any).
\end{lemma}
\begin{proof}
Since there are always actions leading only to $E$, $U(s)=1$ is an invariant as long as there are any actions. The actions can only be removed by \textsc{MakeTerminal}. The only condition when $U(s)\neq 1$ at the beginning of ODQL is for the state $\zero$.
\qed\end{proof}
Due to Lemma \ref{lem:EC-values} we know that when we enter an EC, we always play a uniform strategy on the actions inside the EC and those that leave the EC whenever their $U$ is also $1$. Due to uniformity of the strategy and bounded branching, we visit each state of the EC infinitely often. 

Since we never collapse, we never get stuck in any EC. Hence every action from every visited EC keeps $U$ equal $1$.

Consider the set $\Sigma_U$ of strategies played infinitely often. Then for every EC visited infinitely often the only reachable BSCC's in the Markov chains induced by strategies of $\Sigma_U$ are with $U$ equal $1$, namely consist of the vertex $\one$ (we never get stuck in any other EC than $\one$ or $\zero$). Therefore, the value achieved under any strategy of $\Sigma_U$ is $1$, hence $V$ is $1$ for these EC, finishing the proof of the lemma.
\qed\end{proof}

From Lemma \ref{lem:EC-dichotomy}, it follows that the only MECs visited are with $V$ equal $1$.
Consider an r.v. assigning to each execution an MDP $M''$ by taking $M'$ (or, more precisely, taking an MDP which the random variable $M'$ has as its value)
and replacing 
each MEC $E$ that is never more visited after $j+i$ \textsc{Explore}'s and has value $v$ by a fresh state $s_E$ with $\enab(s_E)=\{a\}$ and $\Delta(s_E,a)(\one) = v$, $\Delta(s_E,a)(\zero) = 1-v$. 

All the ECs of MDP $M''$ have $V$ equal $1$ and the states of $M''$ have the same $V$ as of $M'$ with probability $1-\epsilon_6-\epsilon_7$. Denote $\mathit{Good}$ the set of executions where $V$ on $M''$ (extended to $\mdp$ by ``decolapsing'') is the same as $V$ on $\mdp$. Their probability $\mathcal P(\mathit{Good})$ is at least $1-\sum_{k=3}^7\epsilon_k$, where $\epsilon_3+\epsilon_4+\epsilon_6+\epsilon_7$ can be made arbitrarily low, hence we can consider $\mathcal P(\mathit{Good})=1-\epsilon_5$, i.e. equal to the probability that only ECs were collapsed.

Define DQL' as the EC-free DQL where input is an MDP where each EC (except for state $\zero$) has $V$ equal $1$. This induces a mapping $\mathit{coll}$ between $\mathit{Good}$ executions of ODQL after $j+i$ calls to \textsc{Explore} on $M'$ and runs of the corresponding executions of DQL' on $M''$. Intuitively, as if we restarted the DQL' on the ``almost'' \emph{collapsed} MDP $M''$ with bounds  initialized to ``conservatively improved'' $U_{j+i}$ (and $L_{j+i}$). The mapping is a measure preserving bijection between $\mathit{Good}$ and 
$\mathit{coll}(\mathit{Good})$.

We now show DQL' is a correct extension of DQL on the ``almost'' collapsed MDPs.
\begin{lemma}
DQL' guarantees the same error tolerance $\delta$ as DQL.
\end{lemma}
\begin{proof}
Observe that $M''$ satisfies the Assumption~\ref{ass:correctMDP}. We conclude by Theorem \ref{thm:PAC-1MEC} of Appendix \ref{app:dql}.
\qed\end{proof}

We now show DQL' is a correct extension of DQL with ``conservatively improved'' initialization.
\begin{lemma}
DQL' starting with any conservative initializations of $U$ and $L$ (i.e. $0\leq L\leq V\leq U\leq 1$) guarantees the same error tolerance $\delta$ for the approximation.
\end{lemma}
\begin{proof}
Since on every MDP where the only EC with $V$ different from $1$ is the state $0$ there is a unique fixpoint of the Bellman reachability equations, the value to which DQL' converges also is the lowest fixpoint, therefore we obtain as precise approximation. Moreover, as $U$ and $L$ can only be closer to $V$, the maximal possible number of changes of values $U$ and $L$ as stated in the proof of Theorem \ref{thm:dql} can only be smaller and thus the error tolerance can only be smaller. 
\qed\end{proof}

As a result of the two lemmata, with probability $\mathcal P(\mathit{Good})-\delta_2$ where $\delta_2<\delta$ %
the execution of DQL' is in $\mathit{coll}(\mathit{Good})$ and returns correct $U,L$ on $M''$, which are, moreover, the same on $M'$. 

Therefore, we return correct approximation with probability $\mathcal P(\mathit{Good}) -\delta_1-\delta_2$, where $\mathcal P(\mathit{Good})>1-\delta$.
Since $\delta_1,\delta_2$ are at most the error tolerance of the underlying DQL
the overall error is less than $\delta$ %
if we run the DQL with error tolerance $\delta/4$ i.e. with $\bar\epsilon=\frac{\eps \cdot (p_{\min} / \max_{s \in \states} \enab(s))^{\vert\states\vert}}{12 \vert\states\vert}$ and $m= \frac{\ln(24 \vert\states\vert \vert\act\vert ( 1 + \frac{\vert\states\vert \vert\act\vert}{\bar\eps})/\delta)}{2 {\bar\eps}^2}$.

\end{document}